\documentclass[10pt]{article}
\usepackage[letterpaper,top=2cm,bottom=2cm,left=3cm,right=3cm,marginparwidth=1.75cm]{geometry}
\usepackage[english]{babel}
\usepackage{amsmath,amssymb,amsfonts}
\usepackage{algorithmic}
\usepackage{graphicx}
\pdfoutput=1
\usepackage{hyperref}
\usepackage{authblk} 
\usepackage[dvipsnames]{xcolor}
\usepackage{colortbl}
\usepackage[utf8]{inputenc}
\usepackage[T1]{fontenc}
\usepackage{diagbox} 
\usepackage{hyperref}
\hypersetup{colorlinks=true, urlcolor=blue, 
linkcolor=red, citecolor=blue}
\usepackage{lipsum}
\usepackage{amsthm}

\usepackage{textcomp}
\usepackage{centernot}
\usepackage[braket, qm]{qcircuit}
\usepackage{braket}
\def\BibTeX{{\rm B\kern-.05em{\sc i\kern-.025em b}\kern-.08em
    T\kern-.1667em\lower.7ex\hbox{E}\kern-.125emX}}
\usepackage{multicol}
\usepackage{wrapfig}
\usepackage{tabularx}
\usepackage{tikz}
\usepackage{pgfplots} 
\usetikzlibrary{pgfplots.groupplots}
\usetikzlibrary{intersections}
\pgfplotsset{compat=1.16}
\pgfplotsset{scaled y ticks=false}
\usepgfplotslibrary{fillbetween}
\usepackage{mathtools}
\usepackage{listings}
\usepackage{url}
\usepackage{booktabs}

\usepackage{array}
\newcommand{\PreserveBackslash}[1]{\let\temp=\\#1\let\\=\temp}
\newcolumntype{C}[1]{>{\PreserveBackslash\centering}p{#1}}
\usepackage{comment}
\usepackage{multirow}
\usepackage{flushend}
\usepackage{float}
\usepackage{cleveref}
\usepackage{algorithm2e}
\RestyleAlgo{ruled}

\newtheorem{lemma}{Lemma}[section]

\definecolor{darkmagenta}{rgb}{0.55, 0.0, 0.55}
\usepackage{float}
\usepackage{afterpage}

\title{\textbf{Constraint‑Preserving XY‑Mixers under Trotterized Adiabatic Evolution}}

\author[1,4]{Abhishek Awasthi\thanks{abhishek.awasthi@basf.com}}
\affil[1]{BASF Digital Solutions GmbH, Ludwigshafen am Rhein, Germany}
\author[2,4]{Maximilian Hess}
\author[2,4]{Salome Lomadze\thanks{salome.lomadze@infineon.com}}
\affil[2]{Infineon Technologies AG, Neubiberg, Germany}
\author[3,4]{\\ Francesco B\"ar}
\author[3,4]{Christian Biefel\thanks{christian.biefel@sap.com}}
\affil[3]{SAP SE, Walldorf, Germany}
\affil[4]{Quantum Technology and Application Consortium (QUTAC), Germany \thanks{This work and the present manuscript were conducted within the Quantum Technology and Application Consortium (QUTAC).}}

\date{}
\begin{document}
\maketitle

\begin{abstract}
Constraint handling remains a central challenge for quantum algorithms applied to combinatorial optimization. Standard approaches based on penalty terms increase problem size, distort energy landscapes, and often degrade algorithmic performance. Constraint‑preserving mixers, such as XY‑mixers, provide an alternative by restricting quantum evolution to feasible subspaces, but their practical implementation on gate‑based hardware requires Trotterization, introducing potentially significant approximation errors.
In this work, we systematically study the interplay between constraint‑preserving XY‑mixers and Trotterized Adiabatic Evolution (TAE). We present a detailed theoretical analysis of the origin and scaling of Trotter errors in XY‑mixers and show that the dominant contribution depends on the size and structure of individual constraints rather than on the total problem size. We validate our analysis through extensive numerical simulations on three representative optimization problems: Portfolio Optimization, the Multi‑Car Paint Shop problem, and a Multi‑Commodity Flow problem. For problems containing a single global equality constraint spanning all variables, we observe that Trotter errors significantly impair the performance of XY‑mixers, making standard Pauli‑X mixers more robust under realistic implementations. In contrast, for problems whose constraints decompose into multiple disjoint local blocks, XY‑mixers consistently outperform X‑mixers by several orders of magnitude, even under Trotterized evolution. Our results establish constraint locality as the key criterion for effective use of XY‑mixers and demonstrate that TAE combined with structure‑aware mixer design offers a robust and theoretically grounded alternative to variational quantum optimization methods. This work provides practical guidance for designing quantum algorithms for constrained optimization on near‑term and emerging quantum hardware. Finally, we also provide a mixer Hamiltonian for the TSP-like 2-way-1-hot constraints.
\end{abstract}

\section{Introduction}
Combinatorial optimization problems are ubiquitous across many industries, for example supply chain planning, shop floor scheduling and logistics.
While classical optimization algorithms are often able to approximately solve practical instances of considerable size with great accuracy, many combinatorial optimization problems remain computationally hard. As a result, combinatorial optimization has become an active area of research in quantum computing.

One of the most prominent quantum optimization problems is the Quantum Approximate Optimization Algorithm (QAOA)~\cite{farhi}. It relies on encoding the optimization problem in a Hamiltonian $\mathcal{H}_{f}$ whose eigenvalues correspond to values of the objective function, i.e. $\mathcal{H}_f \ket{x} = f(x) \ket{x}$.
This requires encoding all the relevant information for the optimization problem in an objective function $f:\{0,1\}^n \rightarrow \mathbb{Z}$. The usual workaround to deal with constraints is to add terms to the objective which penalize infeasible assignments.
There are significant problems which come with the penalty term method. First, so-called slack variables are needed to encode inequality constraints, increasing the problem size. Second, the factors of the penalty terms often have to be chosen quite large in order to guarantee that solutions do not violate the initial constraints, which introduces several numerical problems.

A possible remedy to this problem is the use of constraint-preserving mixers~\cite{Hadfield_2019, fuchs_2022}. Recall that QAOA circuits consist of alternating applications of $\exp(-i\gamma\mathcal{H}_f)$ and $\exp(-i\beta\mathcal{H}_M)$, where the typical choice for $\mathcal{H}_M$ is $\mathcal{H}_M = \sum_{i=0}^{n-1} X_i$.
In the broadest sense~\cite{Hadfield_2019}, the idea behind constraint-preserving mixers is to replace $\exp(-i\beta\mathcal{H}_M)$ with a family of unitaries $U(\beta)$ with two properties.
The first is that $U(\beta)\ket{x}$ is a superposition of feasible states whenever $\ket{x}$ is feasible.
The other is that for two feasible states $\ket{x}, \ket{y}$, there exist $r \in \mathbb{N}$ and $\beta \in \mathbb{R}$ such that $\bra{y}U(\beta)^r\ket{x}>0$.
One way to obtain such a family of unitaries is to choose a suitable Hamiltonian $\mathcal{H}_M$ and take $U(\beta)=\exp(-i \beta \mathcal{H}_M)$~\cite{fuchs_2022}.

A concrete example for this technique is given by so-called $XY$-mixers which arise from the Hamiltonian $\mathcal{H}_{XY} = \sum_{(i,j) \in \mathcal{E}} X_i X_j + Y_i Y_j$, where $\mathcal{E} \subset \{0,1,\dots,n-1\}^2$. The $XY$-mixer is constraint-preserving for so-called $k$-hot constraints $\sum_{i} x_i = k$ which are encountered in a multitude of optimization problems, such as portfolio optimization~\cite{He_2023}, traveling salesperson problem (TSP)~\cite{Lin_TSP_1965} or paint shop optimization~\cite{Yarkoni2021}.

In the present work, we explore different subtleties which come up when applying $XY$-mixers to different optimization problems. The first is related to the implementation of $U_{XY}(\beta)=\exp(-i \beta \mathcal{H}_{XY})$. The standard method of realizing the exponential of $\mathcal{H}_{XY}$ is to approximate it via a Trotter decomposition. We find that the error of this approximation grows with the number of summands of $\mathcal{H}_{XY}$, i.e. the size of $\mathcal{E} \subset \{0,1,\dots,n-1\}^2$. This finding is underscored by numerical experiments on two different problem classes, namely the portfolio optimization problem and the binary paint shop problem. While we encounter a single $k$-hot constraint encompassing every problem variable in the portfolio optimization problem, the binary paint shop problem has multiple $1$-hot constraints of only two variables. Our results clearly show the negative effect of the large approximation error in the former case, while in the latter case, an approach with $XY$-mixers clearly outperforms the standard method using $X$-mixers.

Another finding relates to eigenstates of the $XY$-Hamiltonian. As QAOA and related quantum algorithms are inspired by the adiabatic principle, they are expected to work best when the algorithm is initialized in the lowest energy eigenstate among $k$-hot states of the mixer Hamiltonian $\mathcal{H}_{XY}$~\cite{He_2023}. 
Another requirement to the initial state is that it can be efficiently prepared.
We compare two common choices of the index set $\mathcal{E} \subset \{0,1,\dots,n-1\}^2$, namely $\mathcal{E}_{\text{full}} = \{(i,j)\in\{0,1,\dots,n-1\}^2: i<j\}$ (full connectivity) and $\mathcal{E}_{\text{ring}} = \{(i,i+1 \mod n): i\in\{0,1,\dots,n-1\}\}$ (ring connectivity). We find that the full $XY$ Hamiltonian has the Dicke state $\ket{D_k^n} = \frac{1}{\sqrt{\binom{n}{k}}} \sum_{|x|=k} \ket{x}$ as an eigenstate, while $\ket{D_k^n}$ is in general not an eigenstate of the ring $XY$-mixer. Consequently, we focus exclusively on the fully connected $XY$-mixer.

\section{Constraints in optimization problems with XY-mixers}\label{sec:xy-mixers}

Constraints in combinatorial optimization problems come in different forms and their integration into variational quantum algorithms carries varying degrees of difficulty. For the purposes of this work, we examine only linear constraints and make two distinctions with respect to their form. First, an equality constraint for an $n-$variable optimization problem has the form $w^T x = c$ where $w \in \mathbb{Z}^n$ denotes the coefficient vector and $c \in \mathbb{Z}$ denotes the threshold.
Its counterpart is an inequality constraint of the form $w^T x \leq c$. Further, we call a constraint unweighted if $w \in \{0,1\}^n$ and weighted if the coefficients take values other than $0$ and $1$. Standard versions of QAOA and many of its variants cannot handle constraints in a native way. Instead, they rely on an encoding of the optimization problem into a Hamiltonian $\mathcal{H}_{\text{final}} \in \mathbb{C}^{2^n \times 2^n}$ whose eigenvalues and eigenvectors correspond to bit strings and their values with respect to an objective function $f(x)$,~\emph{i.e.},
\begin{equation}
    \mathcal{H}_{\text{final}} \ket{x} = f(x)\ket{x} \text{ for all } x \in \{0,1\}^n\;.
\end{equation}
In the case of $2-$local Hamiltonians, this corresponds directly to formulating the optimization as a Quadratic Unconstrained Binary Optimization Problem (QUBO) with quadratic objective function ${f:\{0,1\}^n \rightarrow \mathbb{Z}}$.
The traditional way to cast a constrained optimization problem into this unconstrained form is to add a penalty term $p(x)$ to the objective function which evaluates to $0$ if the constraint is satisfied and to a high positive value if it is not. This is relatively easy to achieve for equality constraints $w^T x = c$ by choosing $p(x)=P(w^Tx - c)^2$ for $P>0$. If one chooses the penalty factor $P$ adequately, the replacement of the constraint with the penalty term yields an optimization problem with the same optimal solution and value. In contrast, handling inequality constraints $w^T x \leq c$ within the same framework presents is not straightforward. The commonly employed~\cite{Lucas_2014} solution is to introduce $l = \lceil\log_2(c)\rceil + 1$ slack variables $y_0,...,y_{l-1} \in \{0,1\}^n$ in order to transform the inequality into an equality. The new constraint then reads
\begin{equation}
    w^T x + \sum_{k=0}^{l-1} 2^k y_k = c\;.
\end{equation}
Then, one proceeds to cast it into a penalty term as outlined above, ${p(x)=P(w^T x + \sum_{k=0}^{l-1} 2^k y_k - c)^2}$. Although this method yields unconstrained problems whose minima coincide with those of the original problems, the introduction of penalty terms brings a host of practical issues.
As the size of the penalty grows with the violation of the constraint, the contributions of the penalty terms to the objective usually outsize the difference between feasible solutions making it hard for optimization algorithms to find optimal solutions. Further, in the case of inequality constraints the addition of slack qubits lead to an even larger fraction of invalid solutions and further constrain the problem sizes one can tackle on near term devices with a limited number of available qubits.

One possible remedy to the issues caused by penalty terms is to ensure the satisfaction of constraints not in the problem Hamiltonian $\mathcal{H}_{\text{final}}$ but in the mixer Hamiltonian $\mathcal{H}_M$. An example for this technique are so-called $XY$-mixers for unweighted equality constraints (also known as $k$-hot constraints) 
\begin{equation}\label{eq:k-hot-constraint}
    \sum_{i=0}^{n-1} x_i = k\;.
\end{equation}

\subsection{Definition of XY-Mixers}
Using an $XY$-mixer as the initial Hamiltonian in an adiabatic algorithm inherently enforces the $k$-hot constraint by conserving the total number of excited qubits (the Hamming weight of the state). A general $XY$-mixer Hamiltonian is parametrized by a set of qubit pairs $\mathcal{E} \subseteq \{(i,j) : 0 \leq i < j \leq n-1\}$:
\begin{equation}
\mathcal{H}_{XY}(\mathcal{E}) =\sum_{(i,j) \in \mathcal{E}} (X_i X_j + Y_i Y_j)\;.
\end{equation}

The choice of $\mathcal{E}$ determines which qubits interact and with it algebraic properties of the Hamiltonian and circuit complexity. The full connectivity version of the $XY$-mixer Hamiltonian uses all possible qubit pairs:
\begin{equation}
\mathcal{E}_{\text{full}} = \{(i,j) : 0 \leq i < j \leq n-1\}
\end{equation}
yielding:
\begin{equation}
\mathcal{H}_{XY}^{\text{full}}=\sum_{(i,j) \in \mathcal{E}_{\text{full}}}X_i X_j + Y_i Y_j\;.
\label{eq:xy_mixer}
\end{equation}
It commutes with the operator $\sum_i Z_i$ that counts the number of qubits in the $\ket{1}$ state. Starting the evolution in a state with exactly $k$ excitations (a superposition of $k$-hot states) therefore confines the dynamics to the subspace of states with Hamming weight $k$ throughout the adiabatic evolution.
Another way to see that the $XY$-mixer preserves $k$-hot constraints, we observe the action of each summand of $\mathcal{H}_{XY}$ on two qubits.
We have
\begin{equation}
\label{xy-simple eigenstates}
\begin{split}
(X_1X_2+Y_1Y_2) \ket{00} &= (X_1X_2+Y_1Y_2) \ket{11} = 0\;, \\
(X_1X_2+Y_1Y_2) \ket{01} &= \ket{10}+\ket{10} = 2\ket{10} \;,\\
(X_1X_2+Y_1Y_2) \ket{10} &= \ket{01}+\ket{01} = 2\ket{01}\;.
\end{split}
\end{equation}
Consequently, the corresponding unitary $\exp(-i\beta(X_1X_2+Y_1Y_2))$ will act as the identity on $\ket{11}$ and $\ket{00}$ and create superpositions of $\ket{10}$ and $\ket{01}$.
As each summand of $\mathcal{H}_{XY}$ acts on only two qubits, the behavior of a single $XY$-term extrapolates to the action of $U_{XY}(\beta)=\exp(-i\beta\mathcal{H}_{XY})$ which will swap every combination of $\ket{0}$ and $\ket{1}$ qubits in the initial state and thus leave the Hamming weight invariant.

The implicit preservation of the constraint makes an explicit $k$-hot penalty term in the problem Hamiltonian superfluous, and thus lets us circumvent the numerical problems that come with it.

\paragraph*{Implementation of XY-mixers:} In order to implement the mixing operator $\exp(-i \beta \mathcal{H}_{XY})$ on a gate-based quantum computer, it is necessary to employ a Trotterization scheme which approximates $\exp(-i \beta \mathcal{H}_{XY})$ with $\prod_{\forall (i, j) \in \mathcal{E}}\exp(-i\beta(X_iX_j+Y_iY_j))$. The individual exponential terms in the approximated product can be implemented exactly due to commutativity of $X_i X_j + Y_iY_j$.
A quantum circuit which implements $\exp(-i\beta X_i X_j)$ can be expressed using the $\mathrm{CNOT}$ and the rotation gate $R_x$, as shown below in Eq.~\eqref{eq:XY_circuit}.
\begin{equation}\label{eq:XY_circuit}
\exp(i\beta X_i X_j) = CX(i,j) R_x^{(j)}(2\beta)CX(i,j) \;,
\end{equation}

where $CX(i,j)$ denotes the $\mathrm{CNOT}$ gate with qubit $i$ as the control qubit and $j$ as the target qubit. We denote a rotation of the $j$-th qubit around the $X$ axis as $R_x^{(j)}$. For the implementation of $\exp(-i\beta Y_i Y_j)$, we simply need to replace the rotation gate $R_y^{(j)}(2\beta)$~\cite{fuchs_2022}.

\subsection{Eigenstates of XY-Hamiltonians} 
Typically, when $XY$-mixers are used, the quantum system is initialized in a Dicke state, a uniform superposition of $k$-hot basis states.
\begin{equation}
\ket{D_k^n} = \frac{1}{\sqrt{\binom{n}{k}}} \sum_{\substack{x \in \{0,1\}^n\\ |x|=k}} \ket{x}\;.
\end{equation}
It can be easily verified that $\ket{D_k^n}$ is an eigenstate of the $XY$ Hamiltonian. Its eigenvalue is given by $2k(n-k)$. To see this, note that in $\mathcal{H}_{XY} \ket{D^n_k}$, the $XY$ Hamiltonian generates $k(k-n)$ distinct excitation swaps, each such state with weight 2 (due to the action of both $XX$ and $YY$) as seen above in Eq~\eqref{xy-simple eigenstates}. There are $k(n-k)$ such bit strings which explains the eigenvalue. It turns out that a Dicke state is the eigenstate with the lowest eigenvalue for $\mathcal{H}_{XY}^{\text{full}}$, for the $k$-hot basis state.

\begin{lemma}\label{lemma:dicke_eigenstates} 
$\ket{D^n_k}$ is the highest energy eigenstate of $\mathcal{H}_{XY}^{\text{full}}$, for $k = \lfloor \frac{n}{2}\rfloor$. Furthermore, for any $1 \leq k \leq n$, $\ket{D^n_k}$ is the highest energy eigenstate of $\mathcal{H}_{XY}^{\text{full}}$ which lies entirely in the $k$-hot subspace, i.e. the space spanned by $\{\ket{x}: |x|=k\}$.
\end{lemma}
\begin{proof}
It is known that the Dicke state $\ket{D_k^n}$ is an eigenstate of the fully connected $XY$-Hamiltonian
$\mathcal{H}^{\mathrm{full}}_{XY}$ with eigenvalue $2k(n-k)$. We show that no larger eigenvalue exists. We write $\mathcal{H}^{\mathrm{full}}_{XY}$ in the computational basis
$\{\ket{x} : x \in \{0,1\}^n\}$. The matrix is real, symmetric, and has zero diagonal entries. For $x \neq y$,
\begin{equation}
\langle x \lvert \mathcal{H}_{\mathrm{XY}}^{\mathrm{full}} \rvert y \rangle = 
\begin{cases}
2, & \text{if $x$ and $y$ differ by exchanging one $1$ and one $0$}\;, \\
0, & \text{otherwise}\;.
\end{cases}
\end{equation}
Fixing a basis state $\lvert x\rangle$ with Hamming weight $|x|=k$, there are exactly $k(n-k)$ such exchanges. Hence, the absolute row sum corresponding to $\lvert x\rangle$ is $2k(n-k)$. By Gershgorin’s circle theorem~\cite{Gershgorin1931}, every eigenvalue $\lambda$ of $\mathcal{H}^{\mathrm{full}}_{XY}$ satisfies 
\begin{equation}
\lambda \leq \max_{0 \leq k \leq n} 2k(n-k)\;,
\end{equation}
which is achieved for $k=\lfloor n/2 \rfloor$. Since
$\lvert D^n_{\lfloor n/2 \rfloor} \rangle$ realizes this eigenvalue, it is a highest-energy eigenstate.

For the second claim, note that $\mathcal{H}^{\mathrm{full}}_{XY}$ preserves Hamming weight. Let $P_k$ denote the projector onto the $k$-hot subspace $S_k = \mathrm{span}\{\lvert x\rangle : |x| = k\}$, and define
$H_k = P_k \mathcal{H}^{\mathrm{full}}_{XY} P_k$.
Within $H_k$, every row again has absolute sum $2k(n-k)$, implying $\lambda_{\max}(H_k) \le 2k(n-k)$. Since $\lvert D_k^n \rangle \in S_k$ is an eigenstate with this eigenvalue, it is the highest-energy eigenstate in the $k$-hot subspace.
\end{proof}

We have now found the highest energy eigenstates of $\mathcal{H}_{XY}^{\text{full}}$. In QAOA-like protocols we are usually interested in preparing the ground state of our mixing Hamiltonian.
For this reason, we will hereafter use the $XY$-mixer with a $-1$ coefficient, as our mixer Hamiltonian, as shown below in Eq.~\eqref{minus_coeff},
\begin{equation}
    \mathcal{H}_{XY}^{\text{full}} = -1\cdot \sum_{i < j} (X_i X_j + Y_i Y_j)\;.
    \label{minus_coeff}
\end{equation}
It follows from Lemma~\ref{lemma:dicke_eigenstates} that its lowest energy eigenstate in the $k$-hot subspace is given by $\ket{D_k^n}$.

\subsection{Choosing mixer connectivity}
The full $XY$-mixer has quadratic implementation cost in the number of variables. There exist alternatives to the full mixer with a linear cost. One example is the $XY$-Ring mixer, which implements the $XY$-Hamiltonian on a ring topology with periodic boundary condition and restricts the interaction to nearest neighbors with addition of the interaction between the last and the first qubit. The Hamiltonian in question is given by
\begin{equation}
    \mathcal{H}_{XY}^{\text{ring}} = -\sum_{(i,j) \in \mathcal{E}_{\text{ring}}} (X_i X_j + Y_i Y_j)\;,
\end{equation}
where $\mathcal{E}_{\text{ring}} = \{(i, i+1) : 0 \leq i \leq n-2\} \cup \{(n-1,0)\}$.
The ring mixer offers certain practical advantages for the implementation of QAOA. Most notable, the ring topology requires less qubits to be connected to each other than the full mixer, and can be implemented with a linear number of gates~\cite{Leipold_2021}. The $XY$-ring mixer also possesses both key characteristics required for constraint-preserving QAOA: it preserves the feasible subspace and it provides connectivity within this subspace~\cite{PhysRevA.101.012320}. However, as explained in the previous paragraph, it is important for the application of $XY$-mixers that the corresponding Hamiltonians have easy-to-prepare ground states. In this regard, the desirable initial states are Dicke states as efficient circuits for their preparation are known~\cite{Baertschi_2019, Baertschi_2022}. 
In general, Dicke states are not the eigenstates of $XY$-ring Hamiltonians. To see this, we consider an example of $2$-hot Dicke state $|D_2^4\rangle$. We analyze the connectivity structure by examining which states can transition to each other under nearest-neighbor swaps.

We define adjacency under the $XY$-ring Hamiltonian as follows. State $\ket{x}$ is said to be connected to state $\ket{y}$ if there exists $i$ such that $(X_i X_{i+1} + Y_i Y_{i+1})\ket{x} \propto \ket{y}$. Here, addition in the index of $X$ and $Y$ is modulo $n$. For the six $2$-hot basis states, the connectivity is:
\begin{center}
\begin{tabular}{|c|c|c|}
\hline
State & Connected to & Degree \\
\hline
$|1100\rangle$ & $|1010\rangle$, $|0101\rangle$ & 2 \\
$|1010\rangle$ & $|1100\rangle$, $|0110\rangle$, $|1001\rangle$, $|0101\rangle$ & 4 \\
$|1001\rangle$ & $|1010\rangle$, $|0101\rangle$ & 2 \\
$|0110\rangle$ & $|1010\rangle$, $|0101\rangle$ & 2 \\
$|0101\rangle$ & $|1100\rangle$, $|1010\rangle$, $|1001\rangle$, $|0110\rangle$ & 4 \\
$|0011\rangle$ & $|1010\rangle$, $|0101\rangle$ & 2 \\
\hline
\end{tabular}
\end{center}

When $\mathcal{H}_{XY}^{\text{ring}}$ acts on $|D_2^4\rangle$, each basis state receives contributions proportional to its degree. States with alternating patterns ($|1010\rangle$, $|0101\rangle$) have degree 4, while states with adjacent patterns ($|1100\rangle$, $|0110\rangle$, $|1001\rangle$, $|0011\rangle$) have degree 2. This non-uniform connectivity breaks the equal superposition. Therefore, the result is not proportional to $|D_2^4\rangle$, proving that Dicke states are not eigenstates of the $XY$-ring Hamiltonian, despite being eigenstates of the $XY$-full Hamiltonian which has uniform all-to-all connectivity. This mismatch between the natural eigenstates of the ring mixer and the constraint-preserving subspace defined by k-hot constraints makes it a poor fit for our optimization objectives.

Therefore, while the $XY$-ring mixer may offer hardware implementation advantages in certain contexts, we focus exclusively on the $XY$-full mixer for the remainder of this work. In the remainder of the manuscript, we use $\mathcal{H}_{XY}$ and $\mathcal{H}_{XY}^{\text{full}}$ interchangeably, referring to the fully connected $XY$-mixer.

\section{Quantum Optimization Algorithms with XY-mixers}\label{sec:TAE}

Most of the work for circuit model optimization is based on QAOA which has become one of the most promising NISQ‑era optimization methods. Its layered structure allows it to perform well on certain classes of problems~\cite{farhi2019,bravyi}, with experiments showing improved solution quality even on 100+ qubit superconducting processors at moderate depth~\cite{pelofske2026}. Despite this promise, the practical performance of QAOA remains inconsistent and its reliance on variational parameter optimization introduces significant overhead and instability. Parameter landscapes are often highly non‑convex, transferred parameters do not reliably yield monotonic improvement with depth, and classical heuristics can match or outperform QAOA for generic constraint‑satisfaction problems~\cite{limits_qaoa}.

These limitations point to a structural issue: while QAOA is inspired by the adiabatic algorithm, it only approximates adiabatic trajectories when parameters are tuned appropriately, something that becomes increasingly difficult as the number of layers and problem sizes grow. A more controlled alternative is to dispense with the variational loop entirely and instead perform Trotterized Adiabatic Evolution (TAE) on a gate‑based quantum computer using a designed annealing schedule. The Trotterized Adiabatic Approach provides mathematically rigorous convergence guarantees via the Trotter product formula and decomposes a continuous time‑dependent Hamiltonian into short, local gate sequences that align naturally with digital quantum hardware. By choosing an annealing schedule one obtains a systematic and theoretically grounded method for preparing low‑energy states while avoiding the instability and unpredictability of variational training~\cite{effective_inequalities}. A short description of TAE for circuit model quantum computing is described below.

\subsection{Trotterized Adiabatic Evolution (TAE)}\label{sec: trotterized adiabatic evolution}
QAOA is a variational algorithm for approximating solutions to combinatorial optimization problems~\cite{farhi,zhou}. Given an Ising Hamiltonian $\mathcal{H}_{\text{final}}$ whose ground state encodes the optimal solution, QAOA alternates between two unitaries:
\begin{itemize}
 \item the \textit{phase separator} $U(\mathcal{H}_{\text{final}},\gamma)=\exp\left[-i\gamma \mathcal{H}_{\text{final}}\right]$ with $\gamma\in[0,2\pi]$,
 \item the \textit{mixer} $U(\mathcal{H}_{\text{init}},\beta)=\exp\left[-i\beta \mathcal{H}_{\text{init}} \right]$ with $\beta\in[0,\pi]$, where $\mathcal{H}_{\text{init}}$ is an initial Hamiltonian whose ground state $\ket{\psi_0}$ is easy to prepare.
\end{itemize}
Starting from the mixer ground state $\ket{\psi_0}$ (appropriate for minimization) the algorithm applies $p$ alternating layers of these operators, yielding the variational state
\begin{align}
 \ket{\psi_{\boldsymbol{\gamma},\boldsymbol{\beta}}}
 &= U(\mathcal{H}_{\text{init}},\beta_p)\, U(\mathcal{H}_{\text{final}},\gamma_p)\cdots U(\mathcal{H}_{\text{init}},\beta_1)\, U(\mathcal{H}_{\text{final}},\gamma_1)\ket{\psi_0}\;,\\
 &= \prod_{l=p}^{1} U(\mathcal{H}_{\text{init}},\beta_l)\, U(\mathcal{H}_{\text{final}},\gamma_l) \ket{\psi_0}\;,
 \label{eq:qaoa_state}
\end{align}
parameterized by $\gamma_1,\dots,\gamma_p$ and $\beta_1,\dots,\beta_p$.\footnote{Note that the limits of the product operator start from $l=p$ to $1$. This is to depict the fact that the unitary operators with parameters for $l=1$ need to be applied to the initial state first.}
Measuring in the computational basis allows estimation of the cost function as the expectation value $\bra{\psi_{\boldsymbol{\gamma},\boldsymbol{\beta}}}\mathcal{H}_{\text{final}}\ket{\psi_{\boldsymbol{\gamma},\boldsymbol{\beta}}}$ and iterative classical optimization of $(\boldsymbol{\gamma},\boldsymbol{\beta})$ to minimize the expectation value should ideally result in yielding optimal or near-optimal solutions.
The TAE, occasionally referred to as digitized quantum annealing~\cite{Kovalsky_2023,kockum2024lecture,Hegade_2022} approximates continuous adiabatic evolution, governed by an annealing time $\mathcal{T}$ and schedule function ${s:[0,\mathcal{T}] \rightarrow[0,1]}$, by discretizing the time-dependent Hamiltonian
\begin{equation}
\mathcal{H}(t) = (1-s(t)) \mathcal{H}_{\text{init}} + s(t)\mathcal{H}_{\text{final}}\;.
\label{adb1}
\end{equation}
The exact evolution operator is
\begin{equation}
\hat U(\mathcal{T}) = \tau \exp\!\left[-i \int_0^{\mathcal{T}}\mathcal{H}(t)\, dt\right]\;,
\end{equation}
where $\tau$ is the time order operator. This unitary is usually approximated by using $p$ Trotter steps of size $\delta t=\mathcal{T}/p$, where the time used in trotter step $l$ is $t=l\cdot \delta t$. This discretization leads to 
\begin{equation}
\hat U(\mathcal{T}) \approx \prod_{l=p}^{1} \exp\!\left[-i\mathcal{H}(l\cdot \delta t)\,\delta t\right]\;.
\label{adb2}
\end{equation}
Applying a first-order Suzuki--Trotter expansion yields
\begin{equation}
\hat U(\mathcal{T}) \approx
\prod_{l=p}^{1}
\exp\left[-i(1-s(l\cdot\delta t))\mathcal{H}_{\text{init}}\delta t\right]\cdot
\exp\left[-i s(l\cdot\delta t)\mathcal{H}_{\text{final}}\delta t\right]\;.
\label{adb4}
\end{equation}
With the ground state $\ket{\psi_0}$ of $\mathcal{H}_{\text{init}}$ the ground state of $\mathcal{H}_{\text{final}}$ is approximated as $\hat U(\mathcal{T})\ket{\psi_0}$. This form directly parallels QAOA, with  $\beta_l=(1-s(l\cdot\delta t))\delta t$ and $\gamma_l=s(l\cdot\delta t)\delta t$, recovering the QAOA state
\begin{equation}
\ket{\psi_{\boldsymbol{\gamma},\boldsymbol{\beta}}} = \prod_{l=p}^{1} \exp\left[-i\beta_l\mathcal{H}_{\text{init}}\right]\cdot \exp\left[-i\gamma_l\mathcal{H}_{\text{final}}\right] \ket{\psi_0}\;,
\label{qaoa_trotter}
\end{equation}
showing that $p$ Trotter steps are referred to as the number of layers in QAOA. The key distinction is that TAE uses fixed time segments for a given evolution schedule $s$, whereas QAOA treats the parameters variationally and therefore requires classical optimization which brings in an added complexity of optimizing those variational parameters.

\subsection{XY-mixer as Initial Hamiltonian for Adiabatic Evolution}
For optimization problems containing {$k$-hot} constraints, we employ the $XY$-mixer as $\mathcal{H}_{\text{init}}$ and initialize the system in the Dicke state $\ket{D^n_k}$, the lowest-energy eigenstate of the XY Hamiltonian within the $k$-hot subspace (Lemma~\ref{lemma:dicke_eigenstates}). This approach eliminates the need for $k$-hot penalty terms in the problem Hamiltonian $\mathcal{H}_{\text{final}}$. Although the global ground state of $\mathcal{H}_{XY}$ is not generally $\ket{D^n_k}$ for all $k$'s, the $XY$ mixer conserves Hamming weight, rendering each $k$-hot subspace invariant. Adiabatic evolution therefore proceeds entirely within the relevant block, and success depends only on the spectral gap within that sector.

According to the adiabatic theorem generalized to eigenstate evolution, if the system starts in the $j$-th eigenstate of $\mathcal{H}_{\text{init}}$ with a finite energy gap to other eigenstates maintained throughout the interpolation $\mathcal{H}(t)=(1-s(t))\mathcal{H}_{XY} + s(t) \mathcal{H}_{\text{final}}$, it will remain in the $j$-th eigenstate of $\mathcal{H}(t)$ at all times. Here, $j$ corresponds to the lowest eigenstate in the $k$-hot subspace. Therefore, the relevant spectral gap is only the gap \emph{within} the $k$-hot manifold, not the global gap. As long as this gap remains sufficiently large relative to the evolution rate, the system adiabatically evolves from $\ket{D^n_k}$ to the ground state of the $k$-hot subspace of $\mathcal{H}_{\text{final}}$. This approach has been discussed in the literature on constraint-preserving adiabatic algorithms and QAOA mixers~\cite{Hadfield_2019, Wang_2020}.

\subsubsection{Trotter error for XY-mixers}\label{sec:trotter_error}
The fully connected XY-mixer Hamiltonian is constructed as a sum of pairwise exchange (swap) operators: 
\begin{equation}
\mathcal{H}_{XY}= -1\sum_{0\leq i<j\leq n-1} \Big(X_i X_j + Y_i Y_j\Big)\;,
\end{equation}
where $X_i, Y_i$ are Pauli operators on qubit $i$, and the sum runs over qubit pairs $(i,j)$ that define allowable exchanges (typically all pairs for a fully-connected mixer, or a subset for structured mixers). Unlike the simple $X$-mixer (which is $\mathcal{H}_M = \sum_{i=0}^{n-1} X_i$ and consists of single-qubit $X$ rotations that all commute with each other), the XY-mixer $\mathcal{H}_{XY}$ is a sum of many two-qubit terms that generally do not commute with one another. Each exchange term $\mathcal{H}_{ij}=(X_iX_j+Y_iY_j)$ does not commute with other terms that share a qubit in common, e.g. $[\mathcal{H}_{ij}, \mathcal{H}_{ik}]\neq0$ when $j\neq k$. In a fully-connected mixer on $n$ qubits, each qubit participates in $n-1$ pair-terms, and there are $n(n-1)/2$ total terms; almost every term overlaps with many others, leading to a dense non-commutativity graph. As a result, the unitary $U_M(t)=e^{-it \mathcal{H}_{XY}}$ cannot be realized as a single gate; it must be decomposed into a sequence of operations, using Trotterization. However, due to the non-commutativity, the Trotter decomposition is not exact and incurs an error. For Trotterization, one splits $\mathcal{H}_{XY}=\sum_{a} \mathcal{H}_a$ into manageable pieces (for instance, individual pair interactions or sets of commuting interactions) and applies them in sequence:
\begin{equation}
e^{-i \mathcal{H} t} = e^{-i \sum_a \mathcal{H}_a t} \approx \prod_a e^{-i \mathcal{H}_a t} + \mathcal{O}\left(t^2 \sum_{a < b} \| [\mathcal{H}_a, \mathcal{H}_b] \| \right)\;.
\label{eq:trotter_error}
\end{equation}
The magnitude of the Trotter error is governed by the norms of these commutators, in first‑order Suzuki–Trotter approximation, which is on the order of $\mathcal{O}\left(t^2 \sum_{a < b} \| [\mathcal{H}_a, \mathcal{H}_b] \| \right)$, as shown in Eq.~\eqref{eq:trotter_error}. In a worst-case scenario like a fully-connected $\mathcal{H}_{XY}$ on all $n$ qubits, the number of non-commuting term-pairs $(\mathcal{H}_a,\mathcal{H}_b)$ goes to $O(n^2)$. Consider an equality constraint acting of $6$ variables:
\begin{equation}
x_0 + x_1 + x_2 + x_3 + x_4 + x_5 = k \;.
\end{equation}
To preserve this constraint, the XY-mixer must act on all six qubits:
\begin{equation}
\mathcal{H}_M^{\text{global}} = -\sum_{0 \leq i < j \leq 5} (X_i X_j + Y_i Y_j)= -\sum_{0 \leq i < j \leq 5} \mathcal{H}_{ij}\;.
\end{equation}
The trotter error during unitarization of the above Hamiltonian is $\mathcal{O}\left(t^2 \sum_{i < j < k} \| [\mathcal{H}_{ij}, \mathcal{H}_{ik}] \| \right)$. For a fully-connected $\mathcal{H}_{XY}$ on $n$ qubits, the number of non-commuting term-pairs $(\mathcal{H}_{ij},\mathcal{H}_{ik})$ grows rapidly with $n$, leading to $O(n^3)$ total non-commuting pairs for the full mixer.
Now consider the case of separate $k$-hot constraints acting on disjoint sets of variables.
\begin{align}
x_0 + x_1 + x_2 &= k_1\;,\text{ and} \\
x_3 + x_4 + x_5 &= k_2\;.
\end{align}
Each constraint acts on a disjoint subset of qubits. The mixer Hamiltonian becomes:
\begin{equation}
\mathcal{H}_M^{\text{local}} = \mathcal{H}_{012} + \mathcal{H}_{345}\;,
\end{equation}
where,
\begin{align}
\mathcal{H}_{012} &= -\sum_{i<j \in \{0,1,2\}} (X_i X_j + Y_i Y_j)\;, \\
\mathcal{H}_{345} &= -\sum_{i<j \in \{3,4,5\}} (X_i X_j + Y_i Y_j)\;.
\end{align}
Since $\mathcal{H}_{012}$ and $\mathcal{H}_{345}$ act on disjoint qubits, they commute:
\begin{equation}
[\mathcal{H}_{012}, \mathcal{H}_{345}] = 0\;.
\end{equation}
Thus, the total Trotter error in this case is much smaller due to the reduced number of non-commuting terms in each Hamiltonian. In this work we study different problems with XY-mixers showing the effect of Trotter errors for different $k$-hot constraints over varying number of variables in equality constraints.

\section{Studied Combinatorial Optimization Problems} 
In this section we state the combinatorial optimization problems for which we present and discuss results in Section~\ref{sec:results}.
\subsection{Portfolio Optimization}\label{subsec:portfolio_opt}

Portfolio optimization is a fundamental problem in finance, where the goal is to select a combination of assets that balances expected return and risk, see \cite{Markowitz_1952}.
In the following, we discuss the so-called mean-variance portfolio optimization problem, where capital is equally distributed among the selected assets, see also \cite{He_2023}.
Let $n$ be the number of available assets. We consider binary decision variables $x_i \in \{0,1\}$ for $i \in \{0, 1, \dots, n-1\}$, where $x_i = 1$ if asset $i$ is included in the portfolio and $x_i = 0$ otherwise. Let $r_i$ be the expected return of asset $i$, and let $S$ be the covariance matrix of the assets. Additionally we introduce a risk factor $\mu>0$ to balance risk and return. We denote by $k$ the number of assets to be selected. Under the assumption of equal investment in selected assets, the objective is to select a set of assets that maximizes return and minimizes risk.
The corresponding binary optimization model then reads
\begin{align}
    \min_{\mathbf{x}\in \{0,1\}^n} \quad & \mu \mathbf{x}^\top S\mathbf{x} - r^\top \mathbf{x}\;,\\
    \text{s.t.} \quad       & \mathbf{1}^\top \mathbf{x} = k\;.
\end{align}
The objective function approximates the total portfolio variance under the equal-weight assumption. This binary formulation is attractive due to its simplicity and applicability in scenarios with limited investment options or cardinality constraints.

The initial Hamiltonian for the equality constraint is defined as the $XY$-mixer, such that
\begin{equation}
\mathcal{H}_{\text{init}} = -1\cdot \sum_{\forall 0\leq i<j\leq n-1} X_iX_{j}+Y_i Y_{j}\;.
\end{equation}

The final Hamiltonian again contains only the objective term in its Ising form. The initial state is prepared as an equal superposition of all $k$-hot states using the Dicke states.

\subsection{Multi-Car Paint Shop Problem (MCPS)}\label{subsec:mcps}
The Multi Car Paint Shop (MCPS) problem arises in automotive manufacturing, where a fixed sequence of cars enters a paint shop and must be painted with different colors~\cite{Yarkoni2021}. The primary objective is to minimize the number of color switches in the sequence, to reduce setup cost and waste. In practice, the number of cars of each type that must receive each color is predetermined by customer orders. In this work we consider the MCPS problem where the total number of different colors is fixed to two, namely, red and blue. Formally, let
\begin{itemize}
  \item \(n\) be the total number of cars in the given fixed sequence,
  \item \(\mathcal{C} = \{C_0, C_1, \ldots, C_{m-1}\}\) denote the car ensembles (groups of cars with the same features). In particular, the cars in any ensemble $C_j$ do not necessarily appear consecutively in the given fixed sequence.
  \item \(k(C_q)\) denote the number of cars in the ensemble \(C_q\) that must be painted blue, while the remaining $|C_q|-k(C_q)$ cars must be painted red.
\end{itemize}

Let, \(x_i \in \{0,1\}\) denote the binary decision variable for car \(i\), where
  \[
    x_i = 
    \begin{cases}
      1, & \text{if car } i \text{ is painted blue}\;, \\
      0, & \text{if car } i \text{ is painted red}\;.
    \end{cases}
  \]
The optimization problem then reads
\begin{equation}
\begin{split}
  \min_{x\in \{0,1\}^n} \quad & \sum_{i=0}^{n-2} (x_i+x_{i+1}-2x_ix_{i+1}) \;,\\
  \text{s.t.} \quad & \sum_{i \in C_q} x_i = k(C_q), \quad \forall q = 0, \dots, m-1\;.
\end{split}
\label{eq:mcps_constraints}
\end{equation}
The objective term gets a zero for any two adjacent cars assigned the same color, while a different assignment of colors to adjacent cars would fetch a $1$. The equality constraint respects the number of blue (and hence red) cars in each ensemble.

As described above, we optimize the MCPS problem using TAE, where the initial Hamiltonian is defined as the $XY$-mixer for each ensemble $q$, whose $k(C_q)$-hot ground state will always preserve the equality constraint. Formally, we define the initial Hamiltonian for MCPS as,
\begin{equation}
\mathcal{H}_{\text{init}} = - \sum_{q=1}^m\sum_{\substack{i<j:\\ i,j\in C_q}} X_iX_j+Y_i Y_j\;.
\label{mixer:mcps}
\end{equation}
This mixer Hamiltonian will always satisfy the $k(C_q)$-hot constraint over all $q$'s. 

The final Hamiltonian $\mathcal{H}_{\text{final}}$ again contains only the objective term in its Ising form, without any penalty term for the equality constraints. The initial state is prepared using the Dicke states, which is the equal superposition of all the feasible states. Note that in the above equality constraint, for each ensemble, we have a set of variables which are independent to any other ensemble, and thus applying the Dicke state for each ensemble separately, provides us the superposition of all feasible states for the complete problem. We then plug the Hamiltonian terms in Eq.~\eqref{qaoa_trotter}, and optimize the problem by TAE, explained in Section~\ref{sec: trotterized adiabatic evolution}.

\subsection{Multi-Commodity Flow Problem (MCFP)}\label{sec:mcf}
Multi-Commodity Flow problems (MCFP) have applications in economically impactful areas like routing passengers through a transport network or planning production along a supply chain network.
We consider a very specific binary version of the problem in a path based formulation.
We are given a directed graph $G=(V,E)$. Each edge $e\in E$ is assigned with a capacity $u_e$ and a cost $c_e$. There are $K$ different commodities $k\in \{0,1,\dots,K-1\}$, which need to pass through some of the edges from a source $s_k\in V$ to a sink $t_k\in V$ and consume a demand $d_k$ on each edge they flow through. In this version of the problem we pre-compute all simple paths for the source-target pair of each commodity $k$ and denote this set of paths by $\mathcal{P}_k$. 
We use the decision variables

\begin{equation*}
x_{k,p}= \begin{cases}
1, & \text{if commodity $k$ flows through a path $p \in \mathcal{P}_k$}\;,\\
0, & \text{otherwise} \;.
\end{cases}
\end{equation*}
With these in total $n=\sum_k |\mathcal{P}_k|$ decision variables, we want to minimize the objective function
\begin{equation}
\min \sum_{k=0}^{K-1} \sum_{p\in\mathcal{P}_k} \sum_{e\in p}  c_{e} \cdot x_{k,p}\;,
\end{equation}
while fulfilling the capacity constraints
\begin{equation}
\sum_{k=0}^{K-1} \sum_{\substack{p \in \mathcal{P}_k: \\ e \in p}} d_k \cdot x_{k,p} = u_{e}, \quad \forall e\in E\;.
\end{equation}
and the requirement that only one path is used for every commodity
\begin{equation}
\sum_{p\in \mathcal{P}_k} x_{k,p} = 1, \quad \forall k=0,\dots K-1\;.
\end{equation}
The latter equality can again be fulfilled with $XY$-mixer. For each commodity $k$ we define the initial Hamiltonian $\mathcal{H}_{\text{init}}$ as
\begin{equation}
\mathcal{H}_{\text{init}} = - \sum_{k=0}^{K-1} \sum_{\substack{p,p'\in \mathcal{P}_k}}  X_{p}X_{p'}+Y_{p} Y_{p'}\;.
\end{equation}

For this particular problem, we consider the capacity constraints as a penalty term in the objective function which is then converted to a final Hamiltonian in the Ising form. We again use the Dicke states as the initial states for each commodity $k$. It again is important to note that the sets of paths for each commodity are disjoint to the paths of any other commodity. Hence, we can create the complete initial state by applying the Dicke state for each commodity separately.

Our specific version of the MCFP is still NP-hard, even on graphs with polynomially many paths. 
We show this by the following reduction from the well-known \textsc{partition} problem: 
For a multiset of natural numbers S, the task is to find a partition into subsets $S_1$ and $S_2$ such that the sum of the numbers in $S_1$ is equal to the sum of the numbers in $S_2$. 
Given an instance of \textsc{partition}, we can easily construct a graph such that every feasible multi-commodity flow yields a solution to \textsc{partition}.
This graph consists of two nodes $s$ and $t$ and two parallel arcs from $s$ to $t$. Both arcs have unit cost and a capacity of half the sum of the numbers in $S$. For every number in $S$ we define a commodity with demand equal to that number. 
If a feasible multi-commodity flow exists, it distributes the commodities, i.e. numbers from $S$, equally on both arcs, thus giving a solution to \textsc{partition}.

\section{Results} \label{sec:results}
We now present comprehensive results and their comparisons for three different optimization problems, utilizing the $XY$-mixers with TAE, as presented in Section~\ref{sec: trotterized adiabatic evolution}. The influence of constraint-preserving mixers being the major aspect of this work, we compared between the usual $X$ and $XY$-mixers using the TAE algorithm. In all our experiments, we created a superposition of all the feasible states as the initial state for TAE and carried out full-state vector simulations to clearly understand the influences of different mixer operators, avoiding any bias arising due to sampling errors. All our tests were carried out using the sinusoidal scheduling function $s:[0,\mathcal{T}] \rightarrow [0,1]$ given by
\begin{equation}
    s(t) = \sin^2\left(\frac{\pi}{2} \sin^2\left(\frac{\pi t}{2 \mathcal{T}}\right)\right)\;.
    \label{sine_schedule}
\end{equation}
In all our experiments we tested the effectiveness of the mixer Hamiltonians by varying the trotter error with the help of $\delta t$ over different number of trotter steps $p$. The annealing time can be computed as $\mathcal{T}=p\cdot \delta t$. We varied the $\delta t$ values from $0.001$ to $0.9$ over different $p$ values.

\subsection{Results: Portfolio Optimization}~\label{res:portfolio}
This section explains our numerical results for the Portfolio Optimization problem introduced in Section~\ref{subsec:portfolio_opt}.

\begin{figure}[H]
\includegraphics[width=0.98\textwidth]{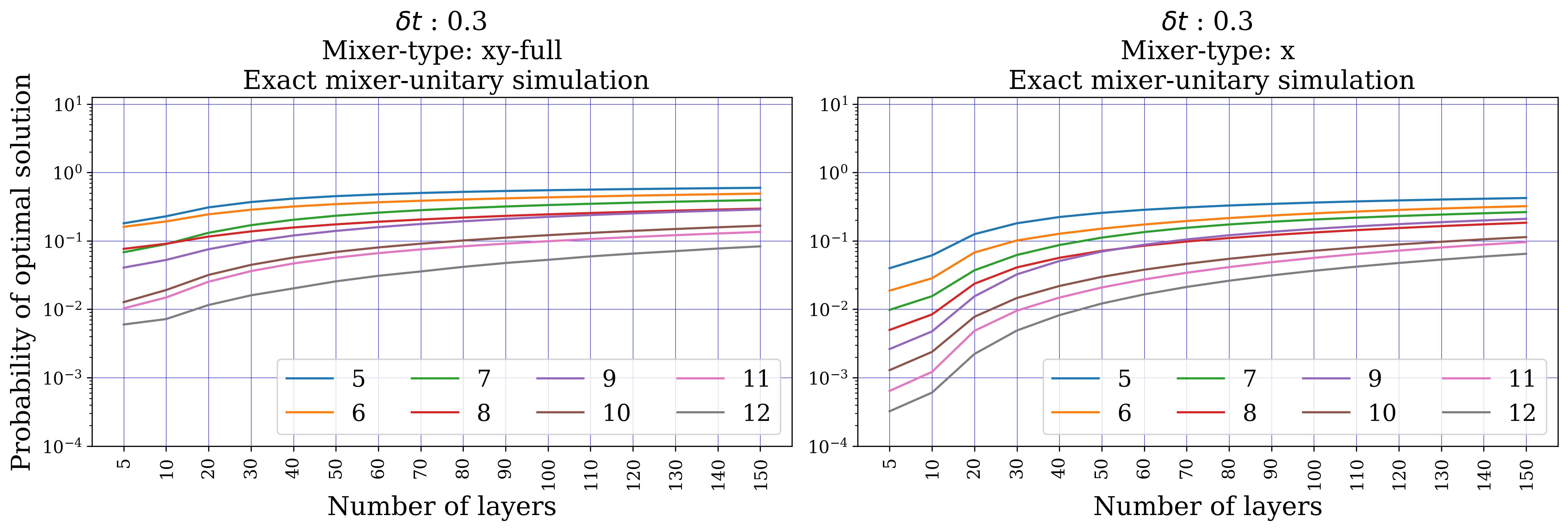}
\includegraphics[width=0.98\textwidth]{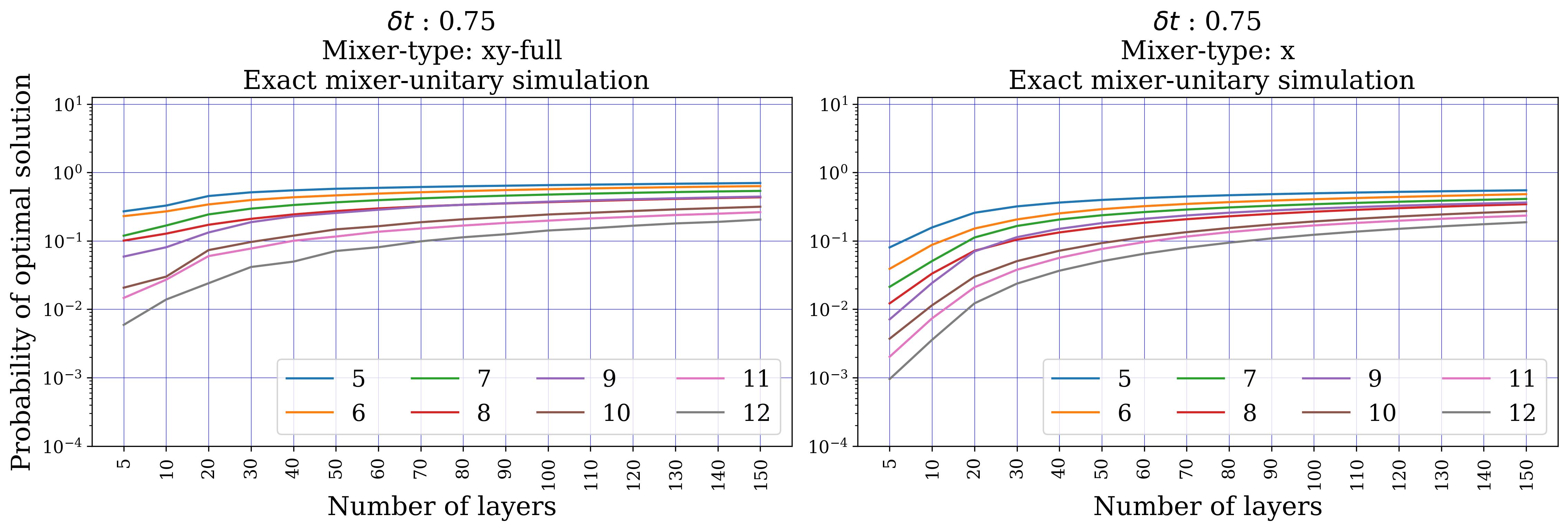}
\caption{\small{This figure demonstrates the breakdown of $XY$ and $X$-mixers performance under TAE. Results for Portfolio Optimization with exact unitary evolution,~\emph{i.e.}, without Trotterization of the mixing and phase separation unitaries but with exact computation of the unitary matrices in Eq.~\eqref{qaoa_trotter}. The plot shows results for $\delta t = 0.3$ and $\delta t = 0.75$, using standard $X$-mixers and the $XY$-mixers with full connectivity. In each plot the $x$-axis shows the number of Trotter steps used in the respective experiment and the $y$-axis shows the probability of measuring the optimal solution which was determined classically beforehand.}}
\label{fig:portfolio_opt_no_error}
\end{figure}
For each problem size of $n\in \{5,6,\dots,11,12\}$ assets we analyzed $10$ randomly generated instances.
The number of assets that a portfolio is allowed to hold varied between $k=\lfloor n/5\rfloor$ and $k=\lfloor 2n/3\rfloor$. The TAE algorithm was executed in two approximation modes. We executed both protocols with the standard $X$-mixers, i.e. $\mathcal{H}_\text{init} = \sum_{i=0}^{n-1} X_i$ as well as for the $XY$-mixers $\mathcal{H}_\text{init}=\mathcal{H}_{XY}^\text{full}$ (full connectivity). Since $XY$-ring can miss some $k$-hot states as explained in Section~\ref{sec:xy-mixers}, we decided not to use $XY$-ring for benchmarking.

In the case of $X$-mixers, the problem's $k$-hot constraint has to be accounted for by adding a penalty term $P\left(\sum_{i=0}^{n-1} x_i - k\right)^2$ to the objective. The factor $P$ has to be chosen such that any profit gained by violating the constraint is negated by the penalty. To this end, we choose ${P=1000}$ for all instances as the maximum return minus the additional risk across all assets does not exceed this number. For the Portfolio optimization problem, we executed TAE with $\delta t$ from the set $\{0.01,0.1,0.2,0.3,0.5,0.75,0.8,0.9\}$ and the annealing scheduling as given in Eq.~\eqref{sine_schedule}. All TAE protocols were run with a number of Trotter steps $p$ from the set on a statevector simulator provided by the quantum information SDK Qiskit~\cite{qiskit_2022}. 

First we simulate the mixer and phase separation unitaries $\exp(-i\beta \mathcal{H}_\text{init})$ and $\exp(-i\gamma\mathcal{H}_\text{final})$ exactly, i.e., without an approximation by Trotterizing, as shown in Figure~\ref{fig:portfolio_opt_no_error}, for $\delta t=0.3$ and $\delta t = 0.75$. All the test results over the mentioned $\delta t$ and $p$ are presented in Appendix~\ref{app_portfolio}. However, it can be clearly seen from Figure~\ref{fig:portfolio_opt_no_error} that increasing $\delta t$ improves the probability of optimal solutions both for $X$ and $XY$-mixers. The reason for this is that increasing $\delta t$ for any given $p$ increases the annealing time and thus provides a longer evolution time. Moreover, the results in Figure~\ref{fig:portfolio_opt_no_error} are obtained by exact simulation of the unitary operators $\exp\left[-i\beta_l\mathcal{H}_{\text{init}}\right]$ and $\exp\left[-i\gamma_l\mathcal{H}_{\text{final}}\right]$ in Eq.~\eqref{qaoa_trotter}, without any trotterization. Thus, increasing the $\delta t$ simply results in higher annealing time. The results without Trotterization showcase the potential of constraint-preserving mixers and especially $XY$-mixers. For both the $\delta t$ values, the $X$-mixer version is outperformed, especially in the regime of low number of Trotter steps $p$. As mentioned above, the results over all the $\delta t$ values can be found in the appendix of this paper.

\begin{figure}[H]
\includegraphics[width=0.98\textwidth]{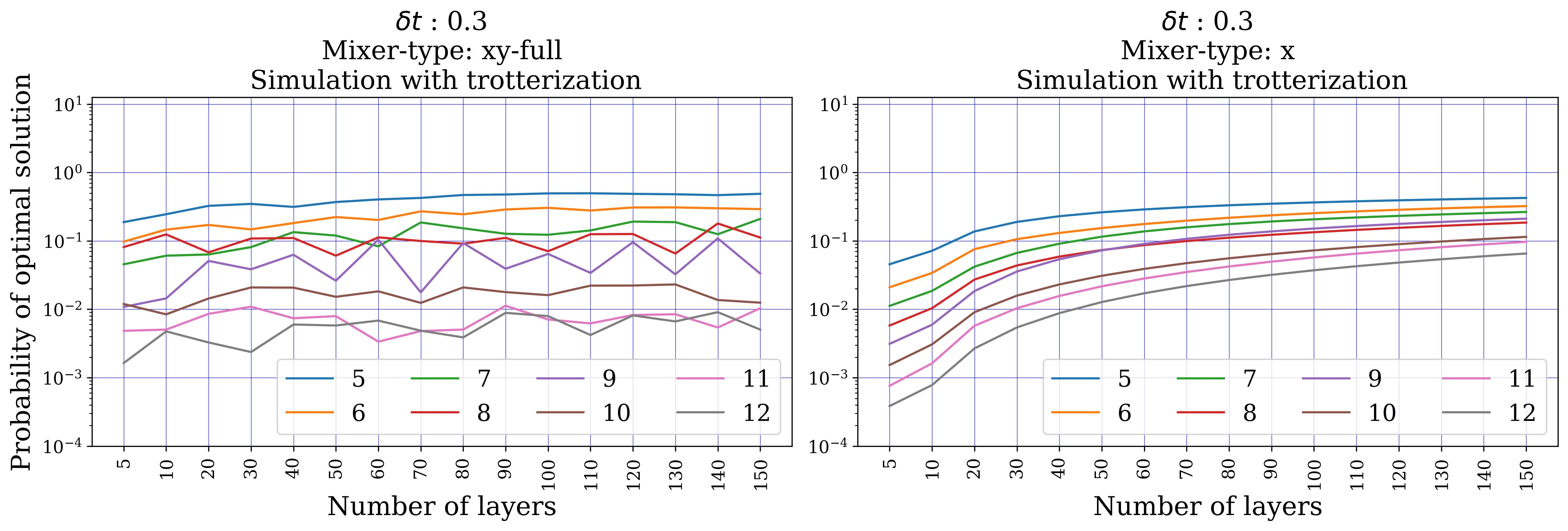}
\includegraphics[width=0.98\textwidth]{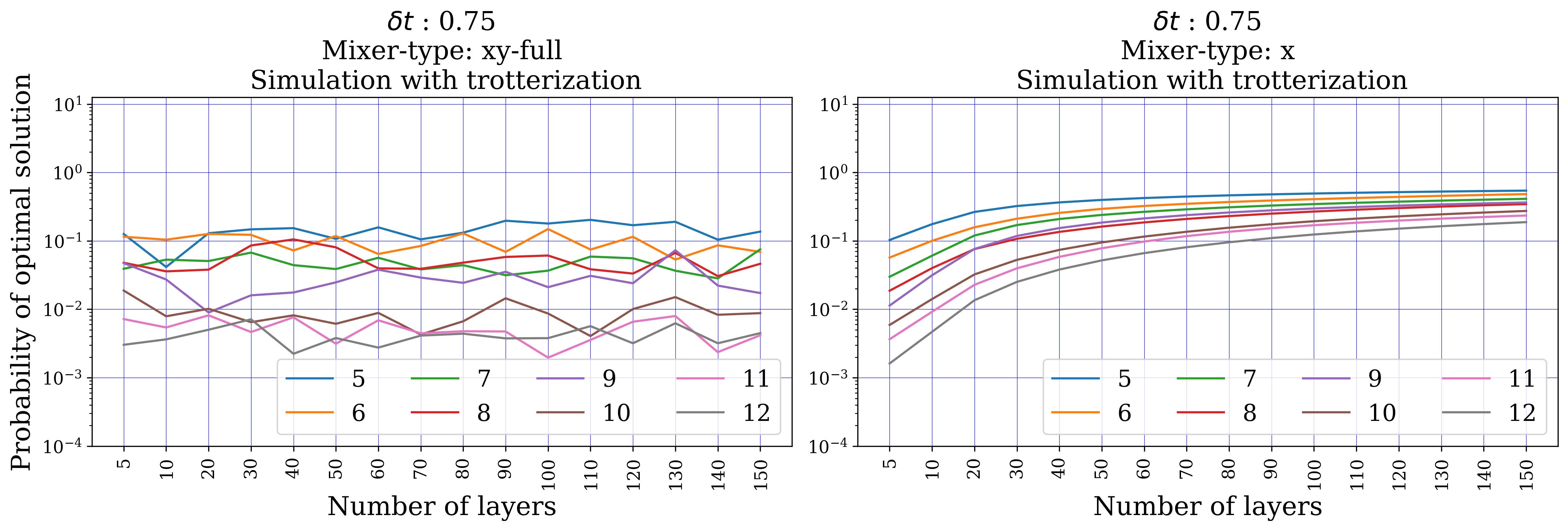}
\caption{\small{Results for Portfolio Optimization with approximated unitary evolution via Trotterization of the mixing and phase separation unitaries. The plot shows results for two different $\delta t = 0.3$ and $0.75$, using standard $X$-mixers as well and the $XY$-mixers with full connectivity. In each plot the $x$-axis shows the number of Trotter steps used in the respective experiment and the $y$-axis shows the probability of measuring the optimal solution which was determined classically beforehand. The probability of optimal solution for any problem size is averaged over $10$ different instances of the same size.}}
\label{fig:portfolio_opt_trotter_error}
\end{figure}

Next, we examine the more realistic mode where mixer and phase separator are approximated by Trotterization as described in Section~\ref{sec:TAE} in order to facilitate their decomposition into primary gates. 
When we test the mixer Hamiltonians with trotterization we observe that the picture changes completely. Especially for the larger problems, the probability of measuring an optimal solution when using the $XY$-mixer no longer grows with the number of Trotter steps and generally shows an erratic behavior, with trotter error playing a major role for such a behavior, and a shift in the quality of results compared to the exact simulation tests in Figure~\ref{fig:portfolio_opt_no_error}. In contrast, the standard TAE version with an $X$-mixer shows a very similar behaviour to that of the perfect simulation scenario. The results are shown in figure~\ref{fig:portfolio_opt_trotter_error}. The reason for this is that the $X$-mixer does not suffer from any trotter error and hence performs similar to an exact simulation. We deduce that the highly erratic behavior of results for the $XY$-mixer is due to the large number of summands in the $XY$-Hamiltonian which lead to a significant error from the Trotter approximation of its exponential. This error propagates even further as the number of Trotter steps increases. We stress that the Trotter error is especially significant for the Portfolio Optimization Problem due to its constraint structure with a single $k$-hot constraint encompassing all problem variables. The reason behind the exact simulation of the unitary operators for the portfolio optimization is to highlight the main issue of trotter error. As explained in Section~\ref{sec:trotter_error}, that for a $k$-hot constraint the trotter error increases as we increase the number of variables in a given equality constraints. Portfolio optimization in specific contains a single equality constraint containing all the problem variables.

In the next sections, we present the results for two more problems where we have more than one equality constraint, with any equality constraint containing only a subset of variables that do not appear in any other equality constraints. Apart from the analyses in Section~\ref{sec:trotter_error}, we will show with the help of some benchmark results that the trotter-error with $XY$-mixers reduces considerably and the results obtained are better than the $X$-mixer by several orders of magnitude. For further tests on other optimization problems we do not carry out tests for exact simulation of unitary matrices, but only the practical way of implementing the $XY$-mixers using Trotterization.

\subsection{Results: Multi-Car Paint Shop (MCPS) Problem }\label{sec:results_mcps}
In this section we present numerical results for the MCPS problem introduced in Section~\ref{subsec:mcps}. The MCPS problem is particularly well suited for evaluating constraint-preserving mixers, as it features multiple disjoint equality constraints, each acting only on a small subset of decision variables corresponding to car ensembles. This structural property fundamentally distinguishes MCPS from the Portfolio Optimization problem studied in Section~\ref{res:portfolio} and has important consequences for the accumulation of Trotter error when using $XY$-mixers. As described in Section~\ref{subsec:mcps}, the MCPS formulation contains a set of independent constraints of the form $\sum_{i \in C_q} x_i = k(C_q)$, with $q = 0,\dots,m-1$, where each ensemble $C_q \subset \{0,\dots,n-1\}$ is disjoint from all others. Accordingly, the mixer Hamiltonian decomposes as a sum of independent local $XY$-mixers, as shown in Eq.~\eqref{mixer:mcps}, where each term acts exclusively on the variables belonging to a single ensemble. As shown in Section~\ref{sec:trotter_error}, such a decomposition significantly reduces the number of non-commuting interaction terms in the mixer Hamiltonian and thereby suppresses the growth of trotter-error.

Representative results for the MCPS problem with problem sizes ranging from $5$ to $20$ qubits (with $10$ different instances for each problem size), are shown in Figure~\ref{fig:mcps_results} for $\delta t = 0.25$ and $\delta t = 0.5$, comparing TAE with $XY$-mixers and $X$-mixers over increasing problem sizes and Trotter depths. In stark contrast to the Portfolio Optimization results of Section~\ref{res:portfolio}, the $XY$-mixer consistently and significantly outperforms the $X$-mixer across all tested problem sizes. This performance advantage persists even in the presence of Trotterization and becomes more pronounced as the number of Trotter steps increases. For moderate to large depths ($p \gtrsim 10$), the probability of sampling an optimal solution using the $XY$-mixer is often several orders of magnitude larger than for the $X$-mixer. Further results for MCPS over several $\delta t$ values for both the $XY$ and $X$-mixers can be found in Appendix~\ref{appendix_mcps}.

\begin{figure}[H]
\includegraphics[width=0.98\textwidth]{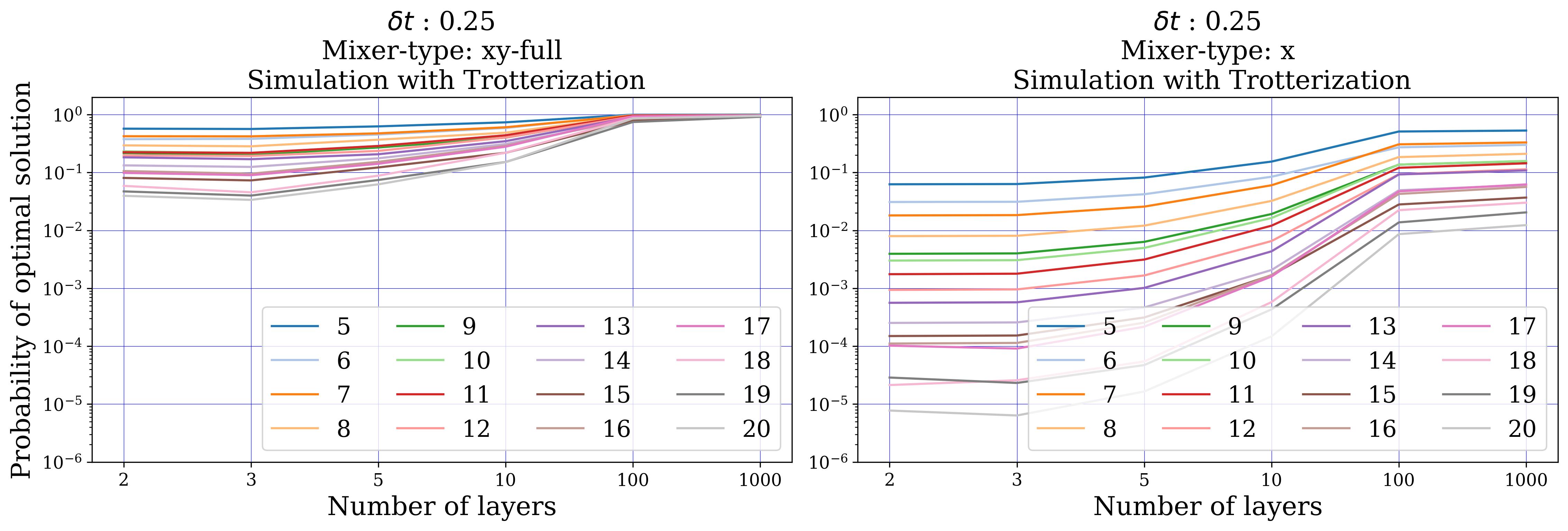}
\includegraphics[width=0.98\textwidth]{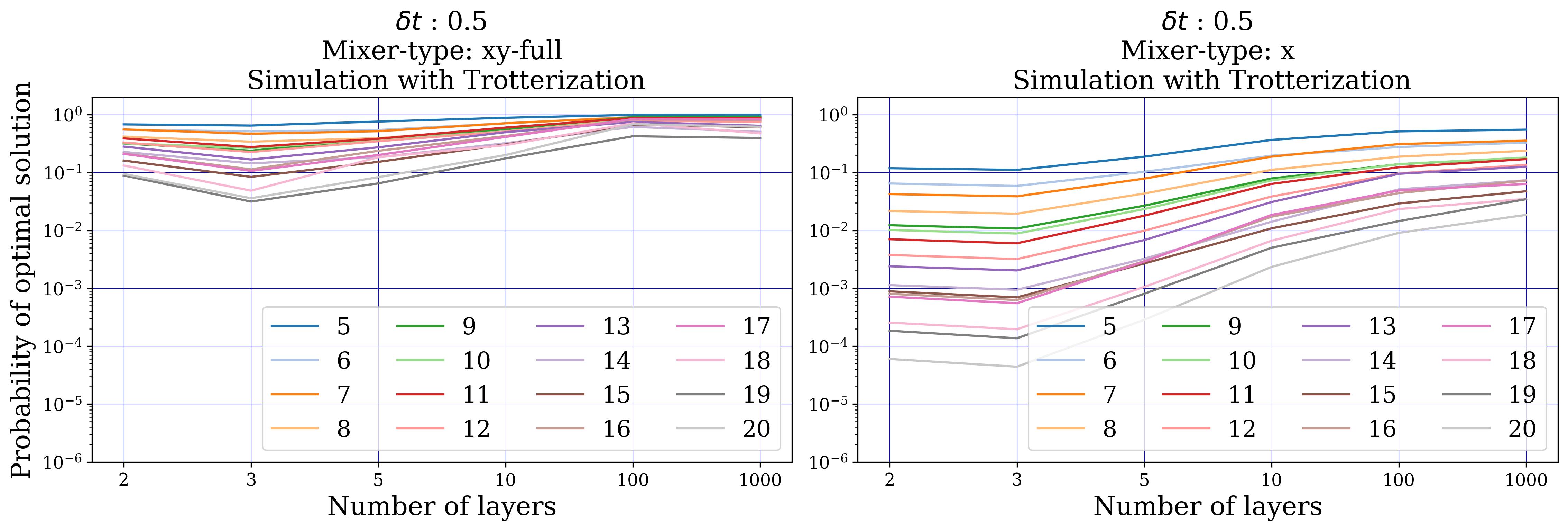}
\caption{\small{Results for Multi-car paint shop problem with approximated unitary evolution via Trotterization of the mixing and phase separation unitaries. The plot shows results for two different $\delta t = 0.25$ and $0.50$, using standard $X$-mixers as well and the $XY$-mixers with full connectivity. In each plot the $x$-axis shows the number of Trotter steps used in the respective experiment and the $y$-axis shows the probability of measuring the optimal solution which was determined classically beforehand. The probability of sampling an optimal solution for any problem size is averaged over $10$ different instances of the same size.}}
\label{fig:mcps_results}
\end{figure}

This behavior can be attributed directly to the local constraint structure of the MCPS problem. Since each equality constraint acts only on a small ensemble, the total mixer Hamiltonian decomposes into commuting blocks, $\mathcal{H}_{\mathrm{init}} = \sum_j \mathcal{H}^{(j)}_{\mathrm{XY}}$, $\bigl[ \mathcal{H}^{(j)}_{\mathrm{XY}},\, \mathcal{H}^{(j')}_{\mathrm{XY}} \bigr] = 0$, for $j \neq j'$. As a result, the number of non-vanishing commutators that govern the leading-order Trotter error, $\mathcal{O}\!\left(\delta t^2 \sum_{a<b} \lVert [\mathcal{H}_a,\mathcal{H}_b] \rVert \right)$, is drastically reduced compared to a single global $XY$-mixer acting on all variables. Moreover, the $XY$-mixer exhibits a smooth and largely monotonic improvement of solution quality with increasing depth, closely resembling ideal adiabatic behavior despite the use of Trotterized evolution. The $X$-mixer, while free of Trotter error in the mixer itself, fails to exploit the problem structure and systematically underperforms.

Overall, the MCPS results demonstrate a regime in which constraint-preserving $XY$-mixers retain their algorithmic advantage even under realistic Trotterization and decisively outperform standard $X$-mixers. Together with the results from Section~\ref{res:portfolio}, this highlights a central message of this work, that the effectiveness of $XY$-mixers under TAE depends critically on the locality and granularity of the enforced constraints.

\subsection{Multi-Commodity Flow Problem (MCFP)}
\label{sec:results_mcmf}
In this section we present numerical results for MCFP introduced in Section~\ref{sec:mcf}. As in the previous sections, our primary objective is to compare the performance of the standard Pauli-$X$ mixer with constraint-preserving $XY$-mixers under TAE.

\begin{figure}[H]
\includegraphics[width=0.98\textwidth,trim={0 3.2cm 0 0},clip]{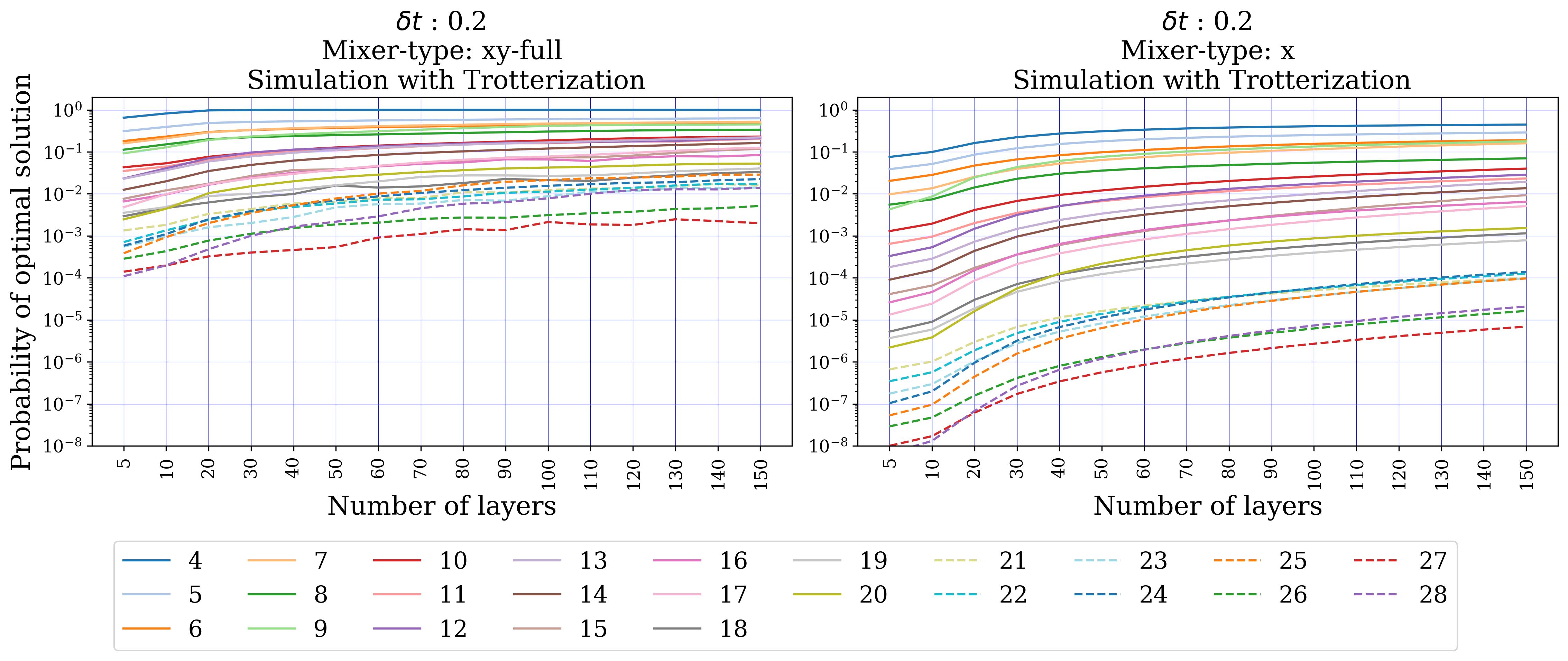}
\includegraphics[width=0.98\textwidth,trim={0 3.2cm 0 0},clip]{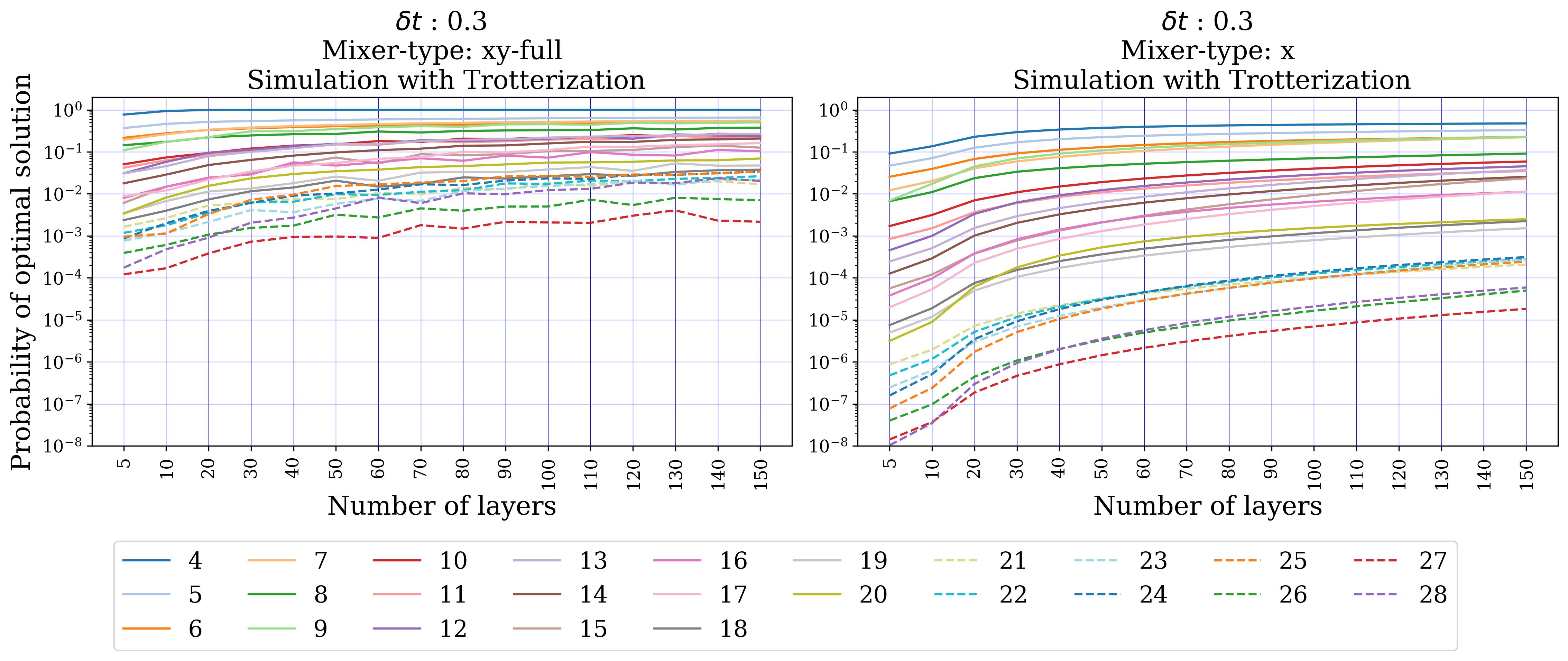}
\includegraphics[width=0.98\textwidth]{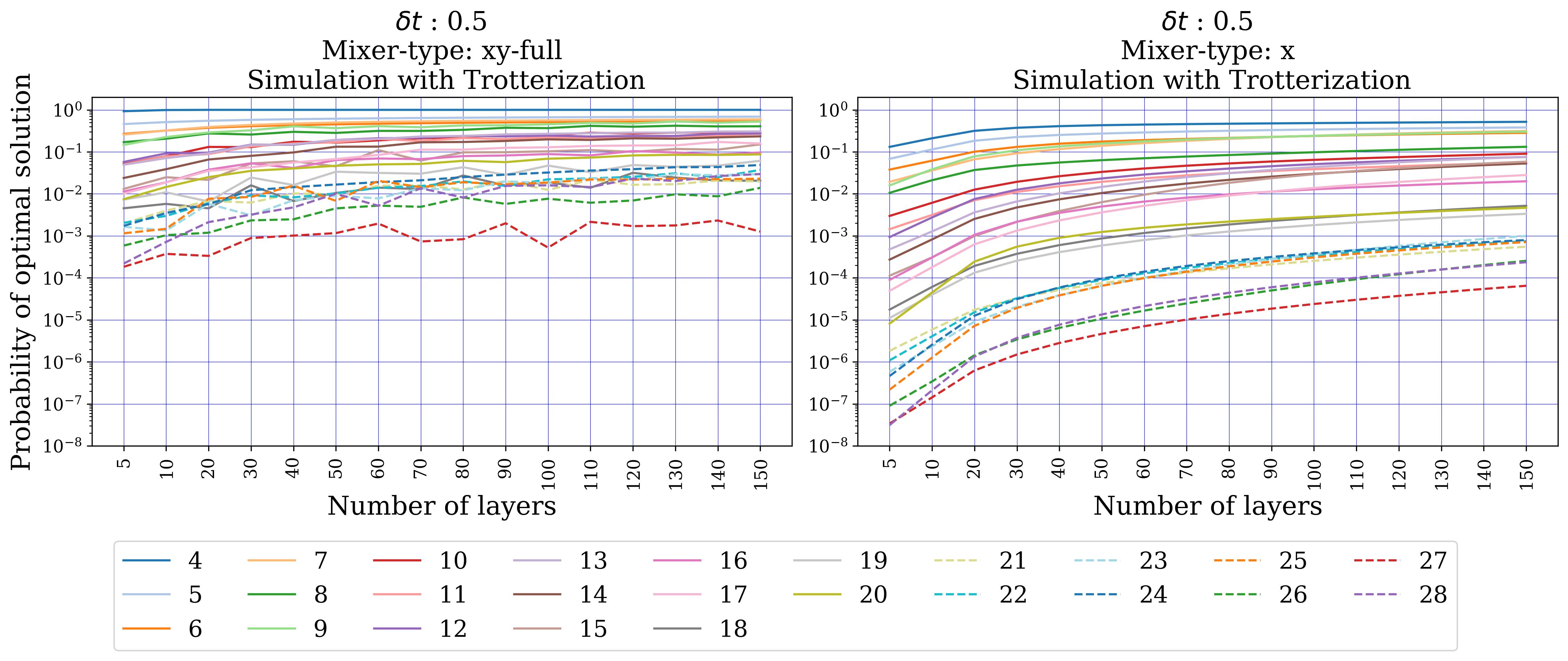}
\caption{\small{Results for the multi-commodity flow problem with approximated unitary evolution via Trotterization of the mixing and phase separation unitaries. The plot shows results for $\delta t = 0.2, 0.3$ and $0.5$, using standard $X$-mixers as well as full $XY$-mixers. In each plot the $x$-axis shows the number of Trotter steps used in the respective experiment and the $y$-axis shows the probability of measuring the optimal solution which was determined classically beforehand. The probability of sampling an optimal solution for any problem size is averaged over $10$ different instances of the same size.}}
\label{fig:mcfp_results}
\end{figure}

The MCFP considered in this work exhibits a constraint structure that lies between the two extremes studied so far, the Portfolio Optimization (Section~\ref{res:portfolio}), which contains a single global equality constraint, and the Multi-Car Paint Shop problem (Section~\ref{sec:results_mcps}), which consists of many small disjoint equality constraints. In the MCFP, each commodity $k$ is associated with one local one-hot constraint, $\sum_{p \in \mathcal{P}_k} x_{k,p} = 1$, where $\mathcal{P}_k$ denotes the set of candidate paths for commodity $k$, and the decision variables $x_{k,p}$ are disjoint across different commodities. As a result, the mixer Hamiltonian decomposes into a sum of independent $XY$-mixers, one for each commodity. Unlike the Portfolio Optimization problem, no equality constraint spans all variables; instead, each constraint acts only on a small subset corresponding to the paths of a single commodity.

We evaluated the performance of TAE for the MCFP on problem sizes ranging from $4$ to $28$ qubits, corresponding to the total number of path variables across all commodities. For each problem size, $10$ random instances were generated. The capacities and demands were drawn from a single-digit integer range, while the penalties were chosen sufficiently large, in this case $10000$. 
We executed TAE with Trotter steps from the set $\{5,10,20,30,\dots,150\}$ and a sinusoidal schedule with $\delta t$ from the set $\{0.01,0.1,0.2,0.3,0.5,0.75,0.8,0.9\}$ on a statevector simulator provided by Qiskit. 

Figure~\ref{fig:mcfp_results} shows representative results for the MCFP with $\delta t = 0.2$, $0.3$, and $0.5$. Across all tested problem sizes and annealing depths, the $XY$-mixer consistently yields a higher probability of sampling an optimal solution than the standard $X$-mixer. The results over all the $\delta t$ values for MCFP are presented in Appendix~\ref{appendix_mcfp}.

While Trotter error is still present due to the non-commutativity of $XY$ terms within each commodity block, its impact is significantly reduced compared to the Portfolio Optimization problem. This can be attributed to the constraint structure in that the $XY$-mixer decomposes into independent blocks for MCFP. The Trotter errors arise only from non-commuting terms within the mixer of an individual commodity. Since the number of paths per commodity is typically small, the number of overlapping two-qubit interaction terms remains limited.

At the same time, the $XY$-mixer rigorously preserves the one-hot constraints throughout the evolution, ensuring that the amplitude remains confined to the feasible subspace for each commodity. In contrast, the $X$-mixer explores the full Hilbert space and relies on large penalty terms in the problem Hamiltonian to discourage infeasible configurations. As observed consistently across all tested instances, this leads to a systematic reduction in the probability mass available for optimal solutions. For moderate Trotter step sizes (notably $\delta t = 0.2$), the $XY$-mixer exhibits stable and monotonic improvement with increasing depth, closely approximating ideal adiabatic behavior. For larger $\delta t$, the effect of Trotter error becomes visible in the form of mild non-monotonicity and saturation, yet the $XY$-mixer continues to outperform the $X$-mixer in all regimes considered.

In summary, the MCFP results confirm that constraint-preserving $XY$-mixers retain a clear performance advantage under TAE when equality constraints act on moderate-size, disjoint subsets of variables. Along with the MCPS results of Section~\ref{sec:results_mcps}, this demonstrates that $XY$-mixers are particularly well suited for structured constrained optimization problems in which feasibility constraints naturally decompose into local blocks.

\subsection{Constraint-Preserving Mixer for the TSP}
\label{sec:tsp_mixer}

Finally, we introduce a special class of constraint-preserving mixers for 2-way-1-hot constraints that is of interest for future research.
The Traveling Salesperson Problem (TSP) seeks a minimum-length Hamiltonian cycle visiting each city exactly once. A standard binary encoding introduces variables
\[
x_{u,t} \in \{0,1\}, \quad u,t \in \{0,\dots,n{-}1\},
\]
where $x_{u,t}=1$ indicates city $u$ is visited at position $t$. The feasible space consists of permutation matrices, enforced by:
\begin{align}
\sum_u x_{u,t} &= 1 \quad \forall t, \label{eq:tsp_time} \\
\sum_t x_{u,t} &= 1 \quad \forall u. \label{eq:tsp_city}
\end{align}

Naive XY-mixers, such as $X_{u,t}X_{v,t} + Y_{u,t}Y_{v,t}$, preserve \eqref{eq:tsp_time} but violate \eqref{eq:tsp_city}, and vice versa for row-wise terms. Thus, independent two-qubit mixers cannot maintain feasibility.

To preserve both constraints, we define a four-qubit \emph{plaquette} mixer that swaps cities $u \ne v$ between positions $t_1 \ne t_2$:
\[
(u,t_1),(v,t_2) \leftrightarrow (u,t_2),(v,t_1).
\]
This is implemented by the term
\begin{equation}
\mathcal{H}^{(u,v)}_{t_1,t_2} =
(X_{u,t_1}X_{v,t_1} + Y_{u,t_1}Y_{v,t_1})
(X_{u,t_2}X_{v,t_2} + Y_{u,t_2}Y_{v,t_2}),
\label{eq:plaquette}
\end{equation}
which acts non-trivially only on feasible configurations. The full mixer is
\begin{equation}
\mathcal{H}_{\mathrm{TSP}} = \sum_{u < v} \sum_{t_1 < t_2} \mathcal{H}^{(u,v)}_{t_1,t_2}.
\label{eq:tsp_mixer}
\end{equation}

Let \( \mathbb{H} = (\mathbb{C}^2)^{\otimes n^2} \) be the Hilbert space of all binary assignments \( x \in \{0,1\}^{n \times n} \). The feasible subspace \( \mathbb{H}_{\mathrm{feasible}} \subset \mathbb{H} \) consists of all basis states corresponding to permutation matrices,~\emph{i.e.}, those satisfying Eqs.~\eqref{eq:tsp_time}--\eqref{eq:tsp_city}. The projector onto this subspace is denoted by
\begin{equation}
P_{\mathrm{feasible}} = \sum_{x \in \mathcal{F}} |x\rangle\langle x|,
\end{equation}
where \( \mathcal{F} \) is the set of feasible configurations. A mixer Hamiltonian \( \mathcal{H}_{\mathrm{TSP}} \) preserves feasibility if it commutes with this projector:
\begin{equation}
[\mathcal{H}_{\mathrm{TSP}}, P_{\mathrm{feasible}}] = 0.
\end{equation}
This ensures that the quantum evolution remains within the feasible subspace when initialized in a valid tour. Trotterized evolution initialized in a feasible state remains within the permutation subspace. Though higher-order, the mixer acts locally on small constraint blocks, and Trotter error is governed by overlapping quartets rather than global interactions.

\section{Conclusion and Future Work}
\label{sec:conclusion}
In this work, we study constraint-preserving $XY$-mixers for combinatorial optimization under TAE, with the goal of understanding when and why such mixers outperform the standard Pauli-$X$ mixer, and under which conditions their practical implementation is limited by Trotterization error. From a theoretical standpoint, $XY$-mixers provide a principled mechanism for enforcing equality constraints by conserving Hamming weight, thereby eliminating the need for penalty terms in the problem Hamiltonian. By initializing the evolution in the $k$-hot Dicke states - the eigenstates of the full-connectivity $XY$ Hamiltonian—adiabatic dynamics confines strictly to the feasible subspace, and algorithmic success depends only on the spectral gap within that subspace rather than on the global gap. This offers a natural alignment between the structure of constrained optimization problems and mixer design.

At the same time, we identify Trotter error as a fundamental practical limitation of $XY$-mixers. Due to the dense non-commutativity of pairwise exchange terms, the Trotter error of a fully connected $XY$ Hamiltonian grows rapidly with the number of variables involved in a single equality constraint. As shown in Section~\ref{sec:trotter_error}, this scaling depends critically on the size of individual constraint blocks rather than on the total problem size.  Our numerical results clearly reflect this distinction. For Portfolio Optimization, which contains a single global $k$-hot constraint spanning all variables, $XY$-mixers perform well only in the idealized setting of exact unitary evolution. Under realistic Trotterized implementations, their performance degrades significantly with increasing problem size and circuit depth, and the standard $X$-mixer becomes more robust despite its reliance on penalty terms. In contrast, for the Multi-Car Paint Shop problem and the Multi-Commodity Flow problem, the equality constraints decompose into multiple disjoint local blocks. In these cases, $XY$-mixers consistently and often dramatically outperform $X$-mixers—even under Trotterization—yielding orders-of-magnitude improvements in the probability of sampling optimal solutions. These results demonstrate that constraint locality substantially suppresses the accumulation of Trotter error and allows the advantages of constraint-preserving evolution to manifest in practice.

We can conclude that the $XY$-mixers are highly effective for structured optimization problems in which equality constraints decompose into small, independent subsets of variables, but they are poorly suited for problems dominated by large global constraints when implemented via naive Trotterization. Beyond the specific results presented here, this work highlights a broader algorithmic perspective. While QAOA remains an important NISQ-era heuristic, its reliance on variational parameter optimization introduces instability, transferability issues, and significant classical overhead. TAE provides a more controlled alternative by combining the theoretical guarantees of adiabatic quantum computation with the flexibility of gate-based implementations, enabling systematic improvements through increased Trotter depth.

Several promising directions emerge from this study. Developing improved Trotterization schemes, such as higher-order decompositions, adaptive ordering, or counterdiabatic terms may substantially extend the range of applicability of $XY$-mixers~\cite{PhysRevLett.133.010603,Hegade_2022}. Exploring partially connected or dynamically sparse $XY$-mixers could further balance constraint preservation with circuit depth, however, as we describe for the $XY$-ring mixer, one should take care that the required $k$-hot states are maintained in the sparse mixer Hamiltonian. Extending these ideas to additional problem classes such as structured scheduling and routing problems, will be essential for establishing structure-aware mixer design as a general paradigm for quantum optimization. Finally, we also present a $XY$-mixer for 2-way-1-hot constraints which is useful for problems like the TSP. Further tests can be carried out to evaluate the performance of this mixer for TSP-like problems to assess the suitability of the $XY$-mixers.

\section*{Acknowledgements}
We would like to thank Lilly Palackal, Janik Sch\"onmeier-Kromer, Karen Wintersperger, Justin Pauckert and Yannick Sch\"afer for their insightful discussions and guidance on technical topics.

\newpage
\appendix
\section{Portfolio Optimization: Results against different $\delta t$ values}\label{app_portfolio}
\subsection{Exact Simulation of Mixer Unitaries: Portfolio Optimization}
\begin{figure}[H]
\centering
\includegraphics[height=1.2\textwidth]{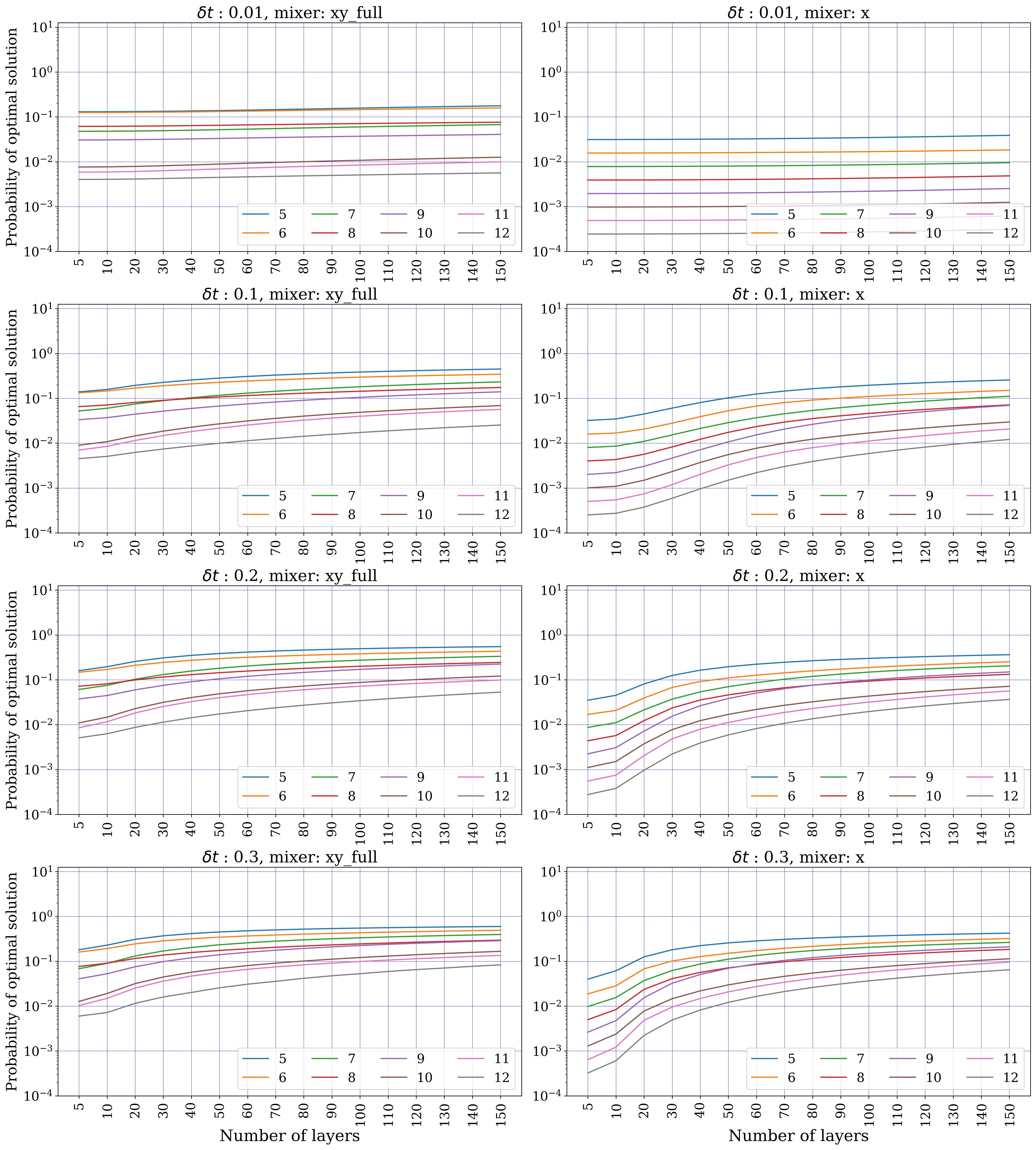}
\caption{\small{Results for Portfolio Optimization with exact unitary evolution,~\emph{i.e.}, without Trotterization of the mixing and phase separation unitaries but with exact computation of the unitary matrices in Eq.~\eqref{qaoa_trotter}. The plot shows results for $\delta t \in \{0.01,0.1,0.2,0.3\}$, using standard $X$-mixers and the $XY$-mixers with full connectivity. In each plot the $x$-axis shows the number of Trotter steps used in the respective experiment and the $y$-axis shows the probability of measuring the optimal solution which was determined classically beforehand.}}
\end{figure}

\begin{figure}[H]
\centering
\includegraphics[height=1.2\textwidth]{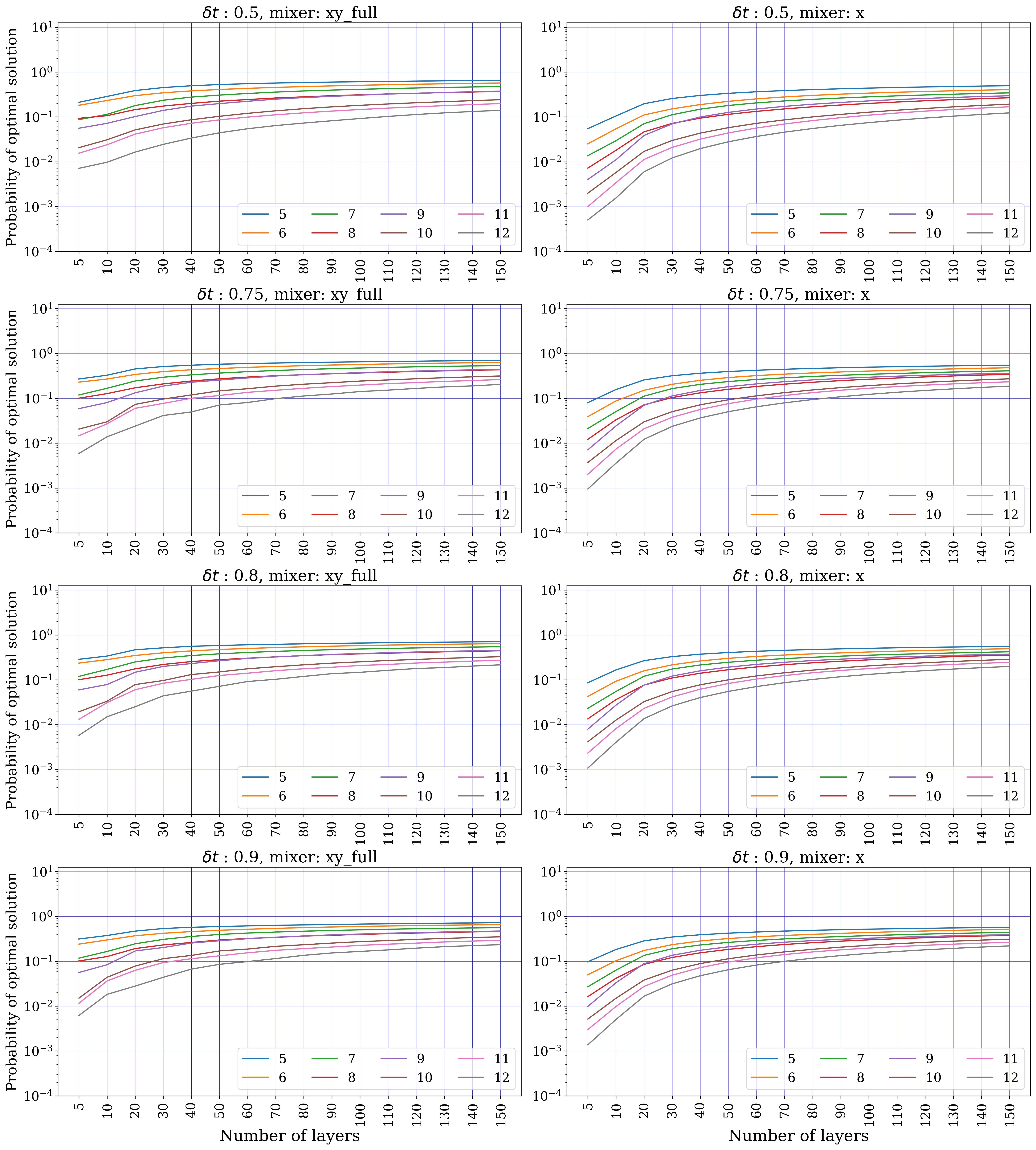}
\caption{\small{Results for Portfolio Optimization with exact unitary evolution,~\emph{i.e.}, without Trotterization of the mixing and phase separation unitaries but with exact computation of the unitary matrices in Eq.~\eqref{qaoa_trotter}. The plot shows results for $\delta t \in \{0.5,0.75,0.8,0.9\}$, using standard $X$-mixers and the $XY$-mixers with full connectivity. In each plot the $x$-axis shows the number of Trotter steps used in the respective experiment and the $y$-axis shows the probability of measuring the optimal solution which was determined classically beforehand.}}
\end{figure}

\subsection{Trotterized Simulation of Mixer Unitaries: Portfolio Optimization}
\begin{figure}[H]
\centering
\includegraphics[height=1.2\textwidth]{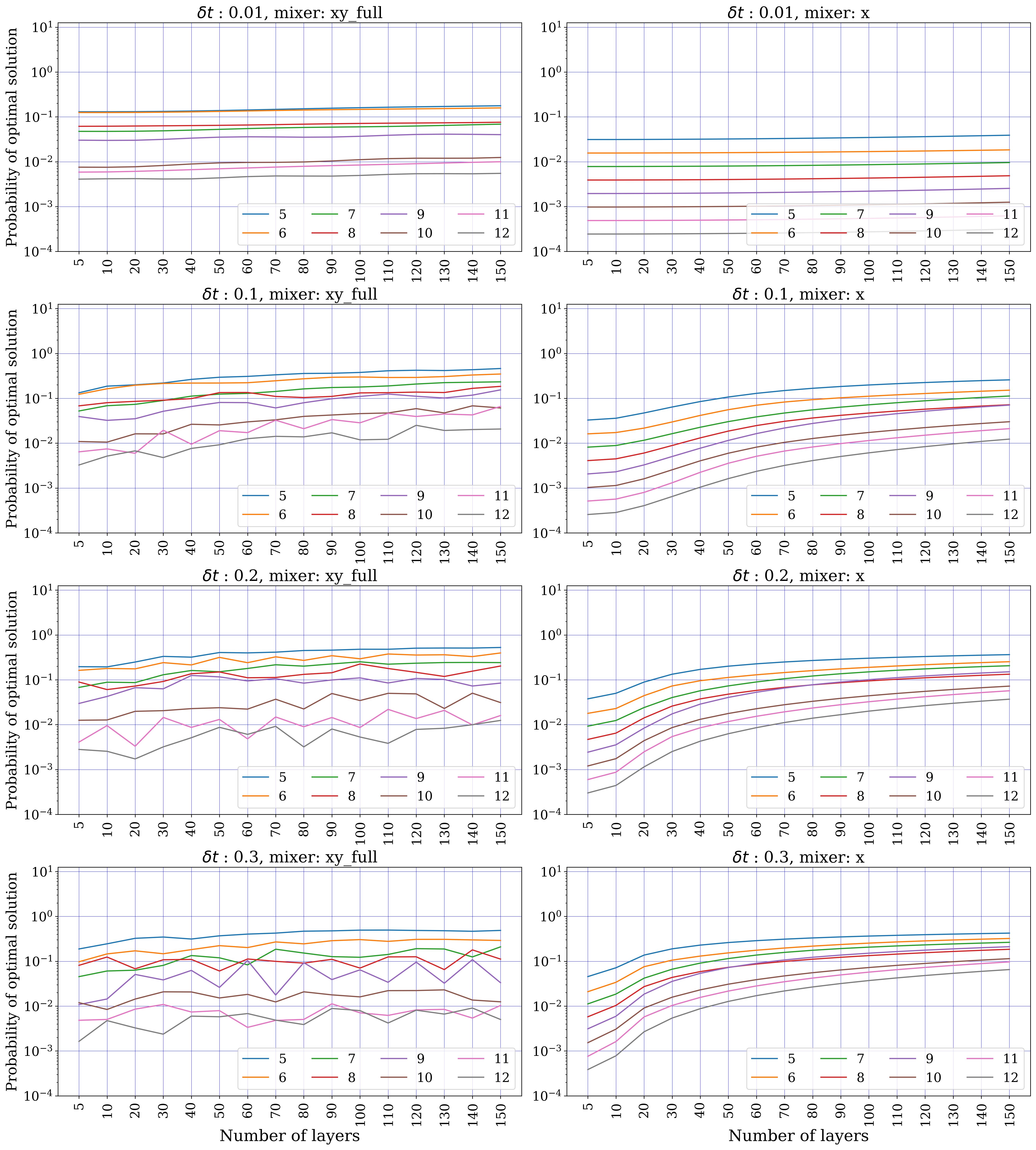}
\caption{\small{Results for Portfolio Optimization with approximated unitary evolution via Trotterization of the mixing and phase separation unitaries. The plot shows results for two different $\delta t \in\{0.01, 0.1,0.2,0.3\}$, using standard $X$-mixers as well and the $XY$-mixers with full connectivity. In each plot the $x$-axis shows the number of Trotter steps used in the respective experiment and the $y$-axis shows the probability of measuring the optimal solution which was determined classically beforehand. The probability of optimal solution for any problem size is averaged over $10$ different instances of the same size.}}
\end{figure}

\begin{figure}[H]
\centering
\includegraphics[height=1.2\textwidth]{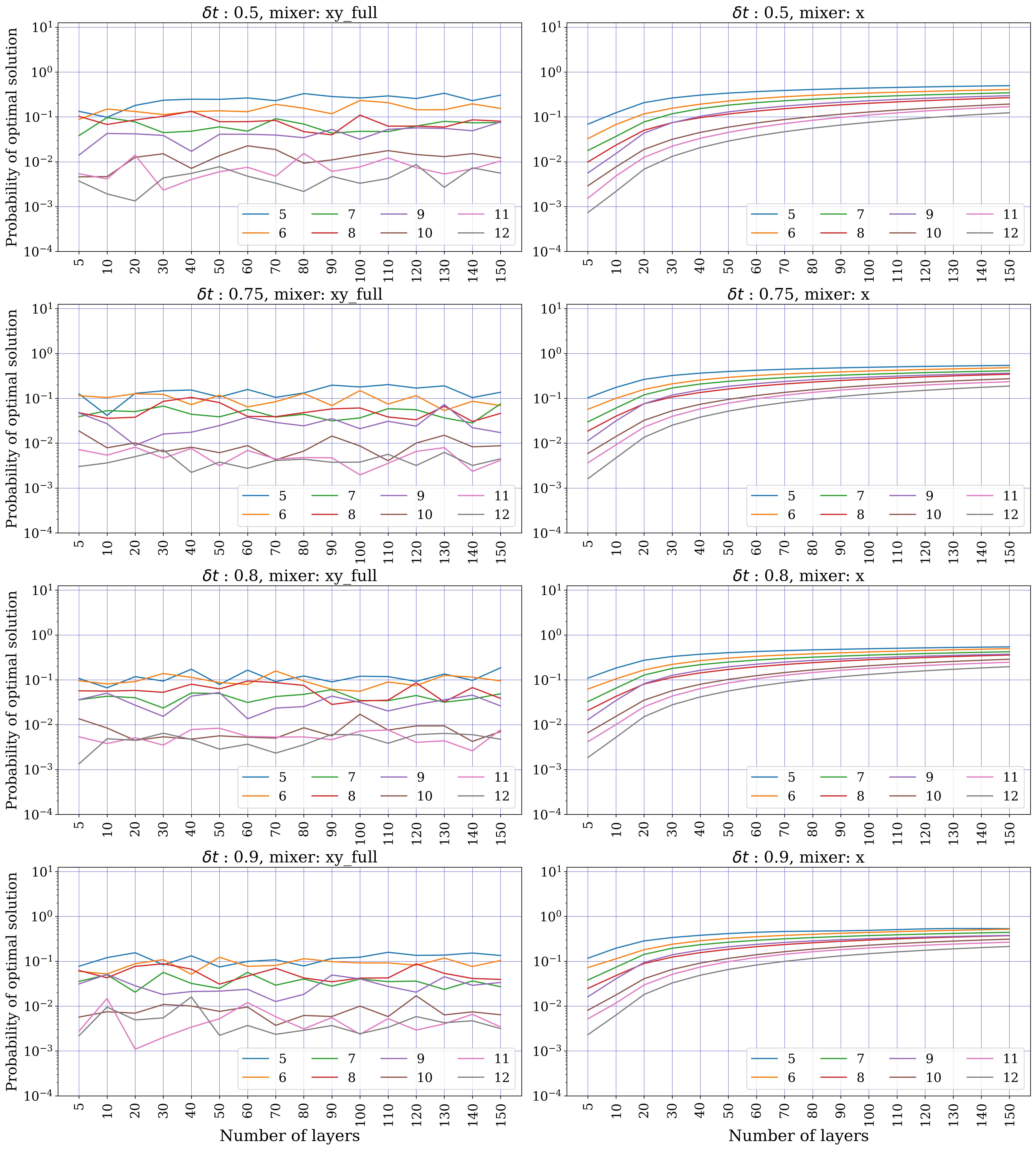}
\caption{\small{Results for Portfolio Optimization with approximated unitary evolution via Trotterization of the mixing and phase separation unitaries. The plot shows results for two different $\delta t \in\{0.5, 0.75,0.8,0.9\}$, using standard $X$-mixers as well and the $XY$-mixers with full connectivity. In each plot the $x$-axis shows the number of Trotter steps used in the respective experiment and the $y$-axis shows the probability of measuring the optimal solution which was determined classically beforehand. The probability of optimal solution for any problem size is averaged over $10$ different instances of the same size.}}
\end{figure}

\section{MCPS Optimization: Results against different $\delta t$ values}\label{appendix_mcps}
\subsection{Trotterized Simulation of Mixer Unitaries: Multi-car Paint Shop}
\begin{figure}[H]
\centering
\includegraphics[height=1.2\textwidth]{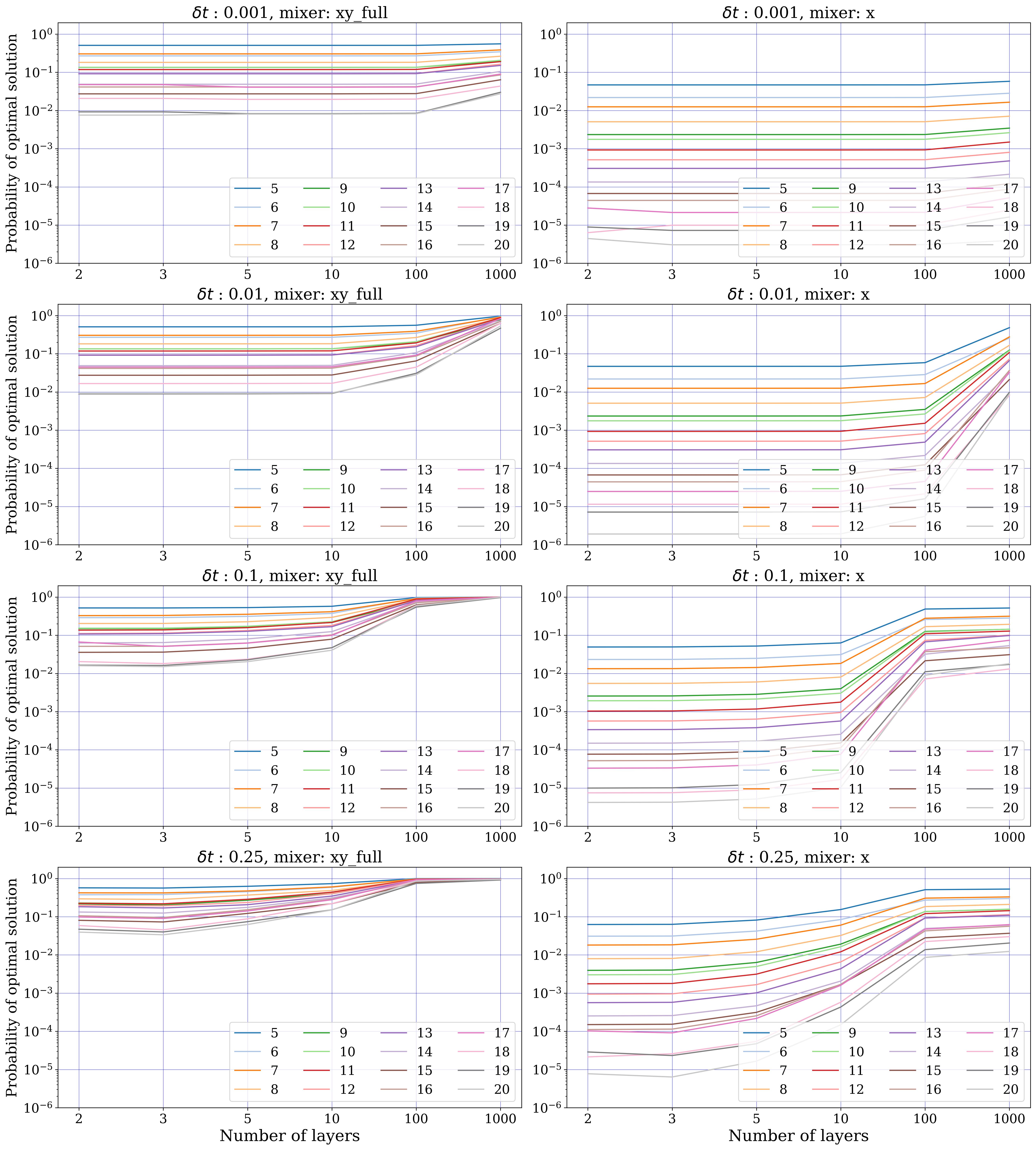}
\caption{\small{Results for Multi-car paint shop problem with approximated unitary evolution via Trotterization of the mixing and phase separation unitaries. The plot shows results for two different $\delta t \in \{0.001, 0.01, 0.1, 0.25\}$, using standard $X$-mixers as well and the $XY$-mixers with full connectivity. In each plot the $x$-axis shows the number of Trotter steps used in the respective experiment and the $y$-axis shows the probability of measuring the optimal solution which was determined classically beforehand. The probability of optimal solution for any problem size is averaged over $10$ different instances of the same size.}}
\end{figure}

\begin{figure}[H]
\centering
\includegraphics[height=1.2\textwidth]{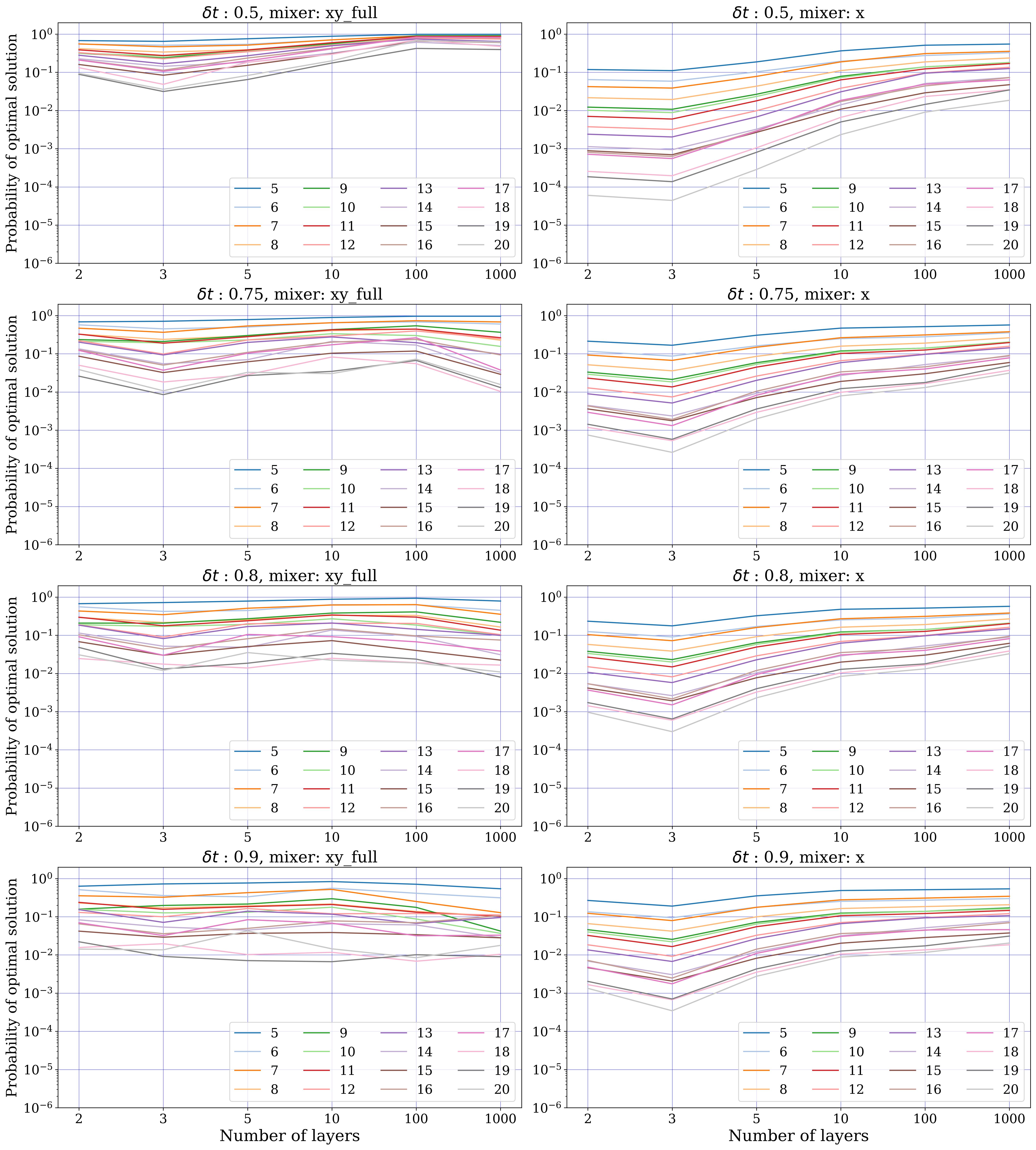}
\caption{\small{Results for Multi-car paint shop problem with approximated unitary evolution via Trotterization of the mixing and phase separation unitaries. The plot shows results for two different $\delta t \in \{0.5, 0.75, 0.8, 0.9\}$, using standard $X$-mixers as well and the $XY$-mixers with full connectivity. In each plot the $x$-axis shows the number of Trotter steps used in the respective experiment and the $y$-axis shows the probability of measuring the optimal solution which was determined classically beforehand. The probability of optimal solution for any problem size is averaged over $10$ different instances of the same size.}}
\end{figure}

\section{MCFP Optimization: Results against different $\delta t$ values}\label{appendix_mcfp}
\subsection{Trotterized Simulation of Mixer Unitaries: Multi-Commodity Flow}
\begin{figure}[H]
\centering
\includegraphics[height=1.2\textwidth]{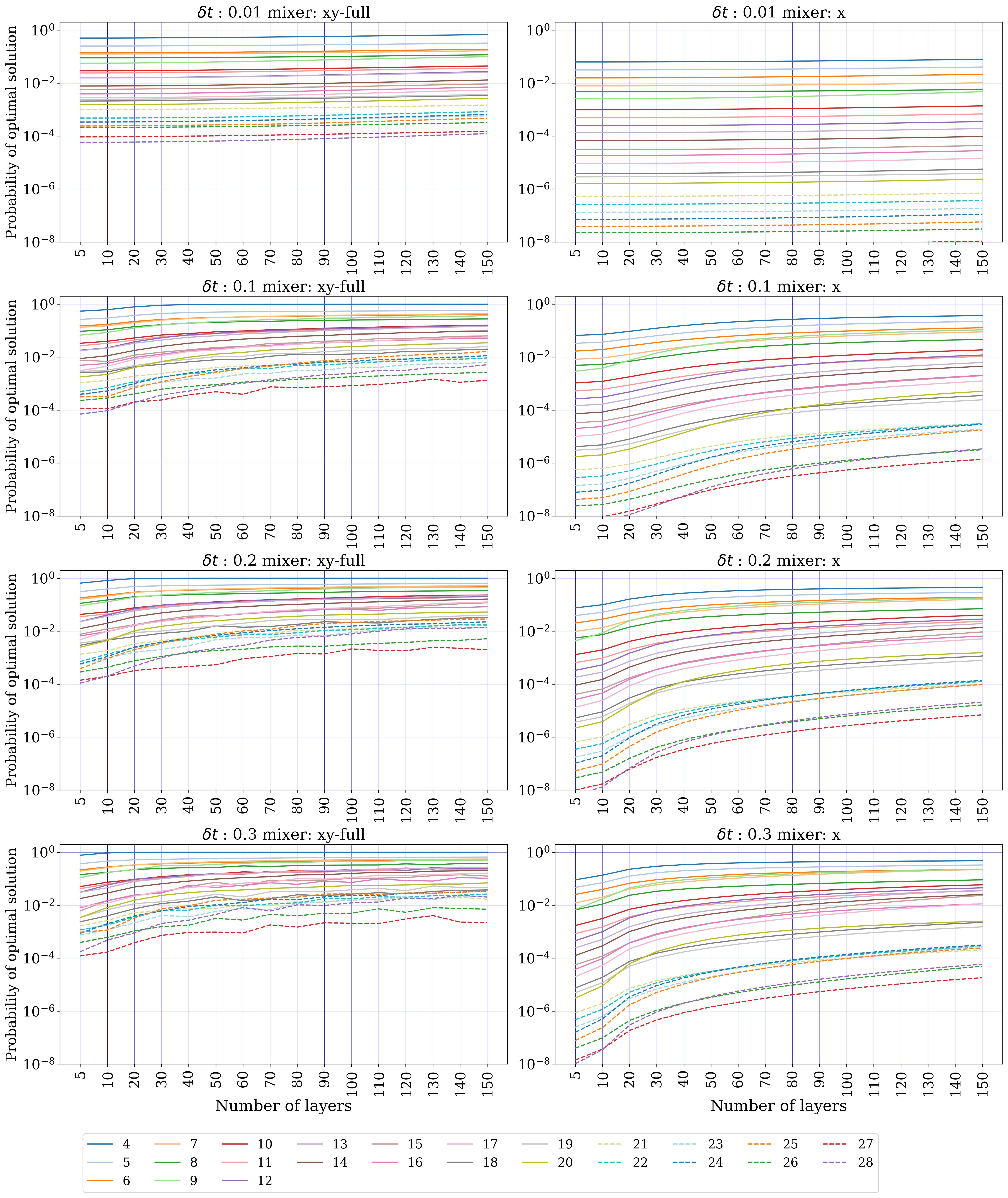}
\caption{\small{Results for the multi-commodity flow problem with approximated unitary evolution via Trotterization of the mixing and phase separation unitaries. The plot shows results for $\delta t \in \{0.01,0.1,0.2,0.3\}$, using standard $X$-mixers as well as full $XY$-mixers. In each plot the $x$-axis shows the number of Trotter steps used in the respective experiment and the $y$-axis shows the probability of measuring the optimal solution which was determined classically beforehand. The probability of optimal solution for any problem size is averaged over $10$ different instances of the same size.}}
\end{figure}

\begin{figure}[H]
\centering
\includegraphics[height=1.2\textwidth]{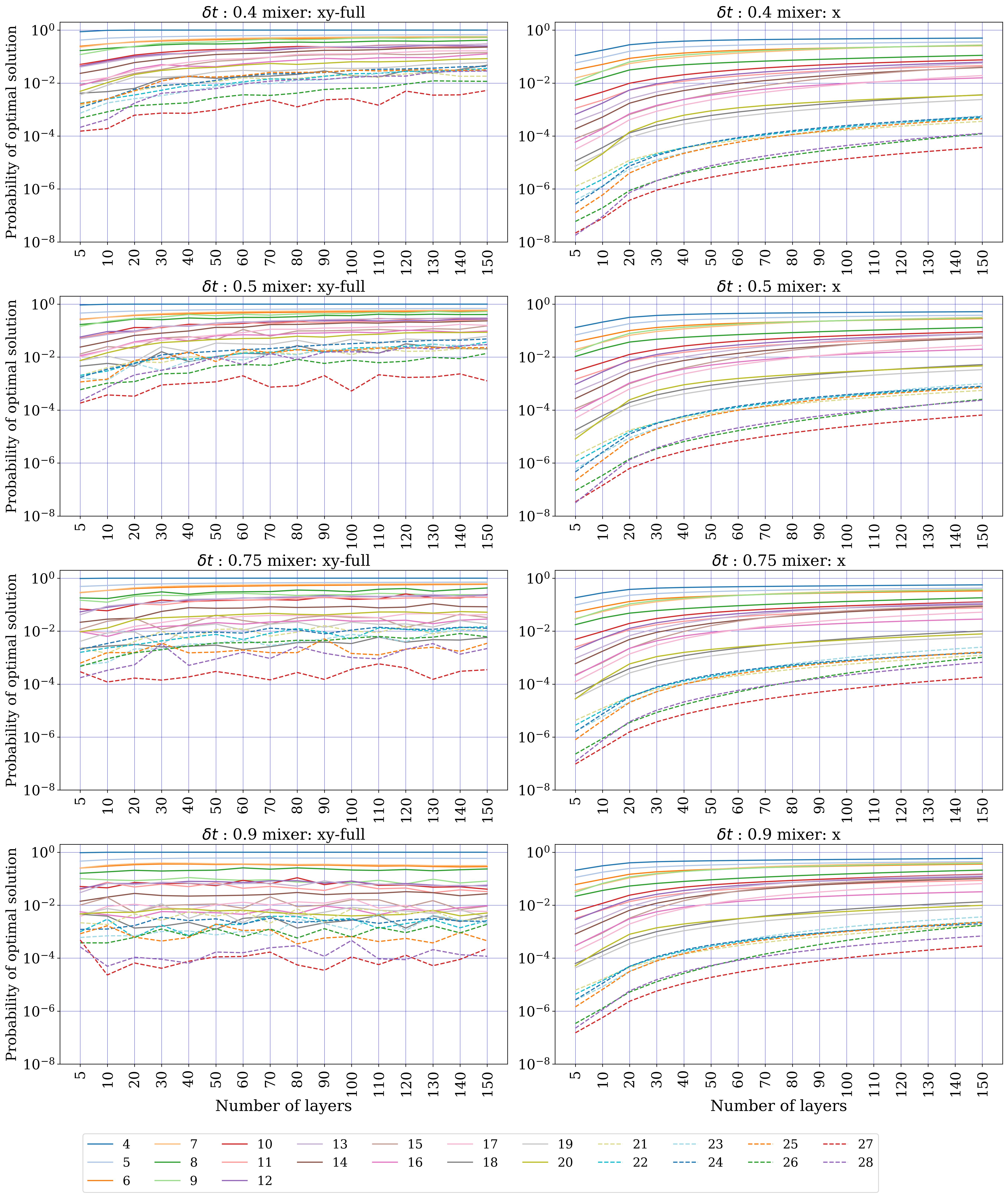}
\caption{\small{Results for the multi-commodity flow problem with approximated unitary evolution via Trotterization of the mixing and phase separation unitaries. The plot shows results for $\delta t \in \{0.4,0.5,0.75,0.9\}$, using standard $X$-mixers as well as full $XY$-mixers. In each plot the $x$-axis shows the number of Trotter steps used in the respective experiment and the $y$-axis shows the probability of measuring the optimal solution which was determined classically beforehand. The probability of optimal solution for any problem size is averaged over $10$ different instances of the same size.}}
\end{figure}

\end{document}